\documentclass[reqno,11pt]{article}
\usepackage[letterpaper,margin=1in]{geometry}
\usepackage[utf8]{inputenc}

\usepackage{amsmath,amssymb,amsthm,amsfonts}
\usepackage{mathtools}
\usepackage{graphicx} 
\usepackage{float}
\usepackage{subfigure} 
\usepackage{arydshln}
\usepackage{multirow}
\usepackage{authblk}
\usepackage{quiver}

\usepackage{bm} 
\usepackage{mathrsfs}   
\usepackage{setspace}

\usepackage{thmtools}
\usepackage{thm-restate}

\usepackage{xcolor}
\usepackage{hyperref}
\hypersetup{
    colorlinks,
    linkcolor={blue!80!black},
    urlcolor={blue!80!black},
    citecolor=teal
}

\usepackage{zref-clever}
\zcsetup{cap,noabbrev,nameinlink}

\newtheorem{theorem}{Theorem}[section]
\newtheorem{lemma}[theorem]{Lemma}
\newtheorem{corollary}[theorem]{Corollary}

\theoremstyle{definition}
\newtheorem{definition}[theorem]{Definition}

\theoremstyle{remark}
\newtheorem{remark}[theorem]{Remark}

\usepackage{url}
\usepackage[title]{appendix}

\usepackage{tikz}  
\usepackage{tikz-cd}  
\usetikzlibrary{positioning}
\usepackage{tikz-qtree}
\usepackage{mathdots}

\usepackage{qcircuit}
\newcommand{\ket}[1]{| #1 \rangle}
\newcommand{\bra}[1]{\langle #1|}
\newcommand{\inner}[1]{\langle #1 \rangle}

\newcommand{\ketbra}[2]{| #1 \rangle \langle #2 |}

\DeclareMathOperator{\poly}{poly}

\newcommand{\on}{\operatorname}
\newcommand{\mb}{\mathbb}

\newcommand{\mc}{\mathcal}

\renewcommand{\Im}{\operatorname{Im}}

\newcommand{\eps}{\varepsilon}

\usepackage{algorithm}
\usepackage{algorithmic}

\makeatletter
\newenvironment{breakablealgorithm}
  {
   \begin{center}
     \refstepcounter{algorithm}
     \hrule height1pt depth0pt \kern3pt
     \renewcommand{\caption}[2][\relax]{
       {\raggedright\textbf{\ALG@name~\thealgorithm} ##2\par}%
       \ifx\relax##1\relax 
         \addcontentsline{loa}{algorithm}{\protect\numberline{\thealgorithm}##2}%
       \else 
         \addcontentsline{loa}{algorithm}{\protect\numberline{\thealgorithm}##1}%
       \fi
       \kern3pt\hrule\kern3pt
     }
  }{
     \kern3pt\hrule\relax%
   \end{center}
  }
\makeatother

\newcommand{\be}{\begin{equation}}
\newcommand{\ee}{\end{equation}}
\newcommand{\bes}{\begin{equation*}}
\newcommand{\ees}{\end{equation*}}
\newcommand{\bea}{\begin{eqnarray}}
\newcommand{\eea}{\end{eqnarray}}
\newcommand{\beas}{\begin{eqnarray*}}
\newcommand{\eeas}{\end{eqnarray*}}
\newcommand{\bal}{\begin{aligned}}
\newcommand{\eal}{\end{aligned}}

\newcommand{\HT}{\mathrm{HT}}
\newcommand{\CT}{\mathrm{CT}}
\newcommand{\ET}{\mathrm{ET}}
\newcommand{\EHT}{\mathrm{EHT}}
\newcommand{\tO}{\widetilde{O}}

\newcommand{\DownTransduce}{
{\rotatebox[origin=c]{270}{$\rightsquigarrow$}}}
\usepackage{bbm}

\title{Elfs, transducers and quantum walks}
\author{Simon Apers\thanks{Universit\'e Paris Cit\'e, CNRS, IRIF, France}, Jérémie Roland\thanks{QuIC, Ecole Polytechnique de Bruxelles, Université libre de Bruxelles, Brussels, Belgium}, Yuxin Zhang\thanks{SKLMS, Academy of Mathematics and Systems Science, Chinese Academy of Sciences,
Beijing, China}}
\date{\today}

\begin{document}

\maketitle

\begin{abstract}
Electric flow sampling (elfs) is a new tool in the quantum walk toolbox and a useful primitive for solving search, sampling and optimization problems on graphs.
We refine this tool by showing that there exists a zero-error transducer for implementing elfs.
More broadly, we establish a zero-error transducer for reflecting about the intersection of two subspaces, yielding an error-free transducer version of the effective gap lemma.
Building on this result, we obtain improved quantum walk algorithms for estimating effective resistances and span program witness sizes with an optimal error scaling, and for sampling from the random walk arrival distribution, via the composition of many elfs.
Using this last algorithm, we obtain an up-to-quadratic quantum speedup for semi-supervised learning on expander graphs.
\end{abstract}

\section{Introduction and summary}

Inspired by the fruitful relationship between random walks, electric networks and linear algebra \cite{doyle1984random,levin2017markov}, a similarly strong relationship between quantum walks, electric networks and linear algebra is currently being uncovered.
Of particular relevance to this work are the contributions of Belovs \cite{belovs2013quantum,belovs2013time} and Piddock \cite{piddock2019quantum}, who showed that quantum walks can be used to prepare (as a quantum state) the electric flow from a source vertex to a sink.
Electric flow sampling (elfs) then corresponds to measuring the resulting quantum state, and this has found applications related to search \cite{belovs2013quantum,piddock2019quantum,li2025multidimensional}, optimization \cite{jeffery2023quantum,wesolowski2024advances} and sampling \cite{apers2022elfs}.
In this work we revisit elfs using the notion of ``transducer'', which is a recently introduced idealization of quantum algorithms \cite{belovs2024taming}.
We show that there exists a zero-error transducer for preparing electric flows.
Building on it, we get improved quantum walk algorithms for estimating effective resistances, improving over \cite{ito2019approximate,piddock2019quantum,apers2022elfs}, and for sampling from the random walk arrival distribution, improving over \cite{apers2022elfs}.
As an example application, we use this last algorithm to speed up the semi-supervised learning algorithm from \cite{zhu2003semi} on expander graphs.

\subsection{Background: graphs, electric networks and quantum walks}
We will consider an undirected graph $G = (V,E,w)$ with edge set $E \subseteq V \times V$ and edge weights $w:E \to \mathbb{R}_{\geq 0}$.
We set $w_{xy}=0$ for $(x,y) \notin E$ so that the degree of a vertex $x\in V$ is $d_x = \sum_{y} w_{xy}$ and $W = \sum_{x \in V} d_x$ denotes twice the total edge weight.
For a source vertex $s \in V$ and sink $M \subseteq V$, the \emph{unit electric flow} from $s$ to $M$ is the unique unit $s$-$M$ flow $f:E \to \mathbb{R}$ that minimizes the dissipated energy
\[
\mathcal{E}(f)
= \frac{1}{2} \sum_{(x,y) \in E} \frac{f_{xy}^2}{w_{xy}} . 
\]
The resulting energy is called the \emph{effective resistance} $R_s = \mathcal{E}(f)$ between $s$ and $M$.
Equivalently, $f$ is the unique unit $s$-$M$ flow that is a potential flow: there exist voltages $v:V \to \mathbb{R}_{\geq 0}$ such that $f_{x y}=w_{x y}\left(v_x-v_y\right)$, where we set $v_y=0$ for all $y \in M$, in which case $v_s=R_s$.

We can derive some relevant random walk quantities from the electric flow.
For a random walk starting from $s$, let the random variable $\tau_M$ denote the time at which the random walk hits $M$.
The \emph{hitting time} $\HT_s$ is then the expectation $\HT_s = \mathbb{E}[\tau_M]$.
The \emph{escape time} $\ET_s$ is the expected time of the last visit to $s$ before $\tau_M$, and the \emph{commute time} $\CT_s$ is the expected time of the first visit back to $s$ after $\tau_M$.
We can express all of these quantities using the voltages \cite{lyons2017probability,apers2022elfs}:
\[
\ET_s
= \frac{1}{R_s} \sum_x v_x^2 d_x
\quad \leq \quad
\HT_s
= \sum_x v_x d_x
\quad \leq \quad
\CT_s
= R_s \sum_x d_x
= R_s W,
\]
where the inequalities follow from the fact that $v_x \leq R_s$ for all $x$.

To define a quantum walk on $G$, we essentially follow \cite{apers2022elfs} (which is inspired by much earlier constructions \cite{szegedy2004quantum,belovs2013quantum}) and introduce the Hilbert space $\ell(E) = \mathrm{span}\{\ket{x,y} : (x,y) \in E\}$.
With $\left|\phi_x\right\rangle=\frac{1}{\sqrt{d_x}} \sum_y \sqrt{w_{x y}}|x y\rangle$ the star state at $x$ and $\mathrm{SWAP}\ket{x,y} = \ket{y,x}$, we define the quantum walk operator $U$ on $\ell(E)$ as
\[
U
= \Big( \sum_{x \notin \{s\}\cup M} 2\ket{\phi_x}\bra{\phi_x} - I \Big) \cdot \mathrm{SWAP}.
\]
Equivalently, $U = (2\Pi_*-I) \cdot (2\Pi_+-I)$ with $\Pi_*$ the projector onto $\mathrm{span}\{\ket{\phi_x}\}_{x \notin \{s\}\cup M}$ and $\Pi_+$ the projector onto $\mathrm{span}\{\ket{x,y} + \ket{y,x}\}_{(x,y) \in E}$.
Notably, the electric flow state
\[
\ket{f} = \frac{1}{\sqrt{2 R_s}} \sum_{(x,y) \in E} \frac{f_{xy}}{\sqrt{w_{xy}}}\ket{xy}
\]
is in the kernel of both $\Pi_*$ and $\Pi_+$, and is therefore left invariant by $U$.
In fact, with $P_{\ker(\Pi_+) \cap \ker(\Pi_*)}$ the projector onto the intersection of these kernels, it crucially holds that
\begin{equation} \label{eq:proj-flow}
P_{\ker(\Pi_+) \cap \ker(\Pi_*)} \ket{\phi_s}
\propto \ket{f}.
\end{equation}
This is exploited in \cite{belovs2013quantum,piddock2019quantum,apers2022elfs,li2025multidimensional}, where they use quantum phase estimation to algorithmically project $\ket{\phi_s}$ into the invariant subspace of the quantum walk $U$, and thus prepare~$\ket{f}$.
Naively, the complexity of this algorithm should depend on the spectral gap of $U$, since that determines the precision required for quantum phase estimation.
In a more careful argument, however, it is possible to use the so-called ``effective gap lemma'' \cite{lee2011quantum} to show that this algorithm can prepare an $\varepsilon$-approximation of $\ket{f}$ in $O(\sqrt{\ET_s}/\varepsilon)$ quantum walk steps~\cite{apers2022elfs}.
In this work, we improve over this approach by showing that there exists a transducer that prepares~$\ket{f}$ exactly.

\subsection{Result 1: effective gap transducers}

A transducer is an exciting new notion in quantum algorithms that was recently introduced by Belovs, Jeffery and Yolcu \cite{belovs2024taming}.
It is an idealized proxy for the usual bounded-error quantum algorithm that can exhibit useful properties such as error-freeness and efficient composition.
In a hot take, Jeffery \cite{jeffery2025composing} even suggested that transducers, as compared to the circuit model, are the more natural model for designing quantum algorithms.

To demonstrate the concept, consider a target transformation
\[
\ket{\phi_{\text{in}}}
\overset{?}{\to} \ket{\phi_{\text{out}}}.
\]
A unitary $U$ is said to transduce $\ket{\phi_{\text{in}}}$ into $\ket{\phi_{\text{out}}}$, denoted by $\ket{\phi_{\text{in}}} \stackrel{U}{\leadsto} \ket{\phi_{\text{out}}}$, if there exists a helper state or ``catalyst'' $w$ that satisfies
\[
\ket{\phi_{\text{in}}} \oplus w
\overset{U}{\rightarrow} \ket{\phi_{\text{out}}} \oplus w.
\]
We call the space where $\ket{\phi_{\text{in}}},\ket{\phi_{\text{out}}}$ live the public space, and that where the catalyst $w$ lives the private space.
While we assume $\ket{\phi_{\text{in}}}$ and $\ket{\phi_{\text{out}}}$ to be normalized, $w$ is not and we call $W = \|w\|^2$ the transduction complexity.
The key insight of \cite{belovs2024taming} is that such a transducer implies the existence of a simple, bounded-error quantum algorithm that makes $O(W/\varepsilon^2)$ calls to $U$ and maps $\ket{\phi_{\text{in}}}$ to an $\varepsilon$-approximation of $\ket{\phi_{\text{out}}}$, \emph{without} knowing the catalyst~$w$.
Despite not having to know it, it is precisely the catalyst that can make transducers error-free.
Finally, by composing these error-free transducers, and only then turning the resulting transducer into a bounded-error algorithm, we get a single global error, rather than an accumulation of each subroutine's individual error.
The right analogy here is with the way randomized Las Vegas algorithms compose, as compared to Monte Carlo algorithms.

The key to our error-free elfs transducer is an error-free transducer version of the aforementioned effective gap lemma \cite{lee2011quantum}.
Since the effective gap lemma has proven a key tool also in other quantum algorithms and quantum query complexity arguments, we expect this error-free version to also find applications beyond elfs.

\begin{restatable}[Effective gap transducer]{theorem}{effgap} \label{thm:eff-gap-transducer}
Consider projectors $\Pi$ and $\Delta$ and quantum walk operator $U = (2\Pi-I)(2\Delta-I)$.
For any $\ket{\psi} \in \ker(\Pi)$ it holds that
\[
\ket{\psi}
\overset{U}{\rightsquigarrow} (2P_{\ker(\Pi) \cap \ker(\Delta)} - I) \ket{\psi}
\]
with public space $\ker(\Pi)$ and catalyst $w = (\Pi-\Pi \Delta \Pi)^+ \Delta \ket{\psi} \in \mathrm{im}(\Pi)$, where $A^+$ denotes the Moore--Penrose pseudoinverse of $A$.
\end{restatable}

We prove this theorem in \zcref{subsec:effective-gap-trans}.
In fact, the same transduction holds (up to a sign) if we just use the reflection $2\Delta-I$ instead of $U$.
We can apply this to electric flows by setting $\Pi = \Pi_*$, $\Delta = \Pi_+$, $\ket{\psi} = \ket{\phi_s}$ and using that $P_{\ker(\Pi_+) \cap \ker(\Pi_*)} \ket{\phi_s} \propto \ket{f}$ (\zcref{eq:proj-flow}).
In this case we show that the catalyst takes the simple form
\[
w
= \frac{1}{R_s \sqrt{d_s}} \sum_{x\neq s} v_x \sqrt{d_x} \ket{\phi_x}.
\]
The effective gap transducer thus implements the reflection
\[
\ket{\phi_s}
\overset{U}{\rightsquigarrow}
(2\ket{f}\bra{f} - I) \ket{\phi_s}
\]
with transduction complexity
\[
W
= \| w \|^2
= \frac{\ET_s}{R_s d_s}-1.
\]
At the same time, we have the trivial (yet important) transduction $\ket{f} \overset{U}{\rightsquigarrow} \ket{f}$ with a zero catalyst.

\begin{remark}
In independent work, Cornelissen~\cite{cornelissen2025quantum} gives an explicit transducer construction for the problem of \emph{span program evaluation}, i.e., deciding for some input $x \in \{0,1\}^n$ whether a given vector $\ket{\psi_x}$ lies in $\ker(\Pi_x) \cap \ker(\Delta_x)$ for projectors $\Pi_x$ and $\Delta_x$.
Similar to our construction, this construction avoids phase estimation and the effective spectral gap lemma.
Our analysis generalizes that in \cite{cornelissen2025quantum} by fully characterizing the induced action of the transducer as a reflection around the kernels intersection.
This is key to our applications, as they require state generation rather than function evaluation.
\end{remark}

\subsection{Result 2: Estimating effective resistances}

By combining it with a reflection about $\ket{\phi_s}$, we can turn the reflection transducer into a transducer for the rotation
\[
(2\ket{\phi_s}\bra{\phi_s} - I) \cdot(2\ket{f}\bra{f} - I),
\]
whose nontrivial eigenvalues are $e^{\pm i 2 \theta}$ with
\[
\sin\theta
= |\inner{\phi_s|f}|
= \frac{1}{\sqrt{2R_s d_s}}.
\]
We can thus use quantum phase estimation to obtain an estimate of $R_s d_s$.
After balancing some terms by adding an extra edge to the graph, similarly (but not quite) to earlier works \cite{belovs2013quantum,apers2022elfs}, we obtain the following result.
It crucially uses the nice composition properties of transducers to implement powers of the rotation.

\begin{restatable}[Effective resistance estimation]{theorem}{effres} \label{thm:intro-eff-res}
Given an upper bound $\bar\ET_s \geq \ET_s$, there is a bounded-error quantum walk algorithm that $\varepsilon$-multiplicatively estimates $R_s d_s$ using
\[
O\left(\sqrt{\bar\ET_s} \left( \frac{1}{\varepsilon} + \log(R_s d_s) \right) \right)
\]
quantum walk steps in expectation.
\end{restatable}

The algorithm and its analysis are given in \zcref{sec:eff-res-known, subsec:eff-res-unknown}.
For comparison, naively combining quantum phase estimation with the original effective gap lemma yields a scaling with the precision as $1/\varepsilon^2$.
Through a more careful argument, balancing out some terms, the state-of-the-art improved this scaling to $1/\varepsilon^{3/2}$ \cite{ito2019approximate,piddock2019quantum,apers2022elfs}.
We strictly improve this scaling to $1/\varepsilon$, also for the generalized problem of estimating the \emph{witness size} of a span program as considered in \cite{ito2019approximate}.
We prove that this $1/\varepsilon$-scaling is in fact optimal.

\subsection{Result 3: Preparing elfs}

The rotation $(2\ket{\phi_s}\bra{\phi_s} - I) (2\ket{f}\bra{f} - I)$ effectively implements quantum amplitude amplification, and so we can also use it to prepare the elfs state $\ket{f}$.
To obtain a high-precision algorithm, we use a variation called fixed-point amplitude amplification \cite{yoder2014fixed} that makes calls to the generalized rotation $((1-e^{-i \varphi}) \ket{\phi_s}\bra{\phi_s} - I) ((1-e^{i \phi})\ket{f}\bra{f} - I)$ for varying $\varphi,\phi$.
While the resulting transducer is no longer error-free, it does improve the error scaling for preparing elfs from polynomial to logarithmic.
Assuming a constant-factor approximation of~$R_s d_s$, we prove the following theorem.

\begin{theorem}[Informal] \label{thm:elfs-informal}
Given an upper bound $\bar\ET_s \geq \ET_s$, there exists a transducer $V$ that makes one call to the quantum walk operator and $\tO(1)$ other operations so that
\[
\ket{\phi_s}
\stackrel{V}{\leadsto} \ket{f} + O(\varepsilon)
\]
with transduction complexity $O\left(\sqrt{\bar\ET_s} \log(1/\varepsilon)\right)$.
\end{theorem}

More precisely, the transducer uses the quantum walk on a modified graph $\hat G$ and maps $\ket{\phi_\sigma}$ to $\ket{\hat f}+O(\varepsilon)$, where $\ket{\hat f}$ is the electric flow state from the added source $\sigma$ to $M$.
See \zcref{sec:high-prec-elfs} for the construction and analysis.
In a second approach, we use a version of the Las Vegas quantum amplitude amplification algorithm in \cite{boyer1998tight}, which is zero error but has unbounded (worst case) runtime.
Correspondingly, the transducer uses an infinite counter (similar to e.g.~\cite{belovs2025space,jeffery2025qma}) and introduces garbage (unlike \cite{belovs2025space,jeffery2025qma}).
We obtain the following transducer.

\begin{theorem}[Informal] \label{thm:intro-elf-inf-bar}
Given an upper bound $\bar\ET_s \geq \ET_s$, there exists a transducer $\hat V$ with an infinite counter that makes one call to the quantum walk operator and $\tO(1)$ other operations so that 
\[
\ket{\phi_s} \ket{0}
\stackrel{\hat{V}}{\leadsto} \ket{f}\ket{\Gamma_s},
\]
with transduction complexity $\tO\left(\sqrt{\bar\ET_s}\right)$.
\end{theorem}
This follows from \zcref{thm:sweep}.

\subsection{Result 4: Elfs process}
\label{sec:elfs-intro}
As our next contribution (and initial motivation), we compose our elfs transducer to simulate the elfs process, as introduced by Apers and Piddock \cite{apers2022elfs}:
\begin{quote}
\textbf{Elfs process:}\newline
From an initial source vertex $s$ and a fixed sink $M \subseteq V$, repeat the following:
\begin{enumerate}
\item
Let $f$ denote the electric flow from source $s$ to sink $M$.
Sample an edge $e$ with probability proportional to $f_e^2/w_e$.
\item
Pick a random endpoint $x$ from $e$ and set the source $s = x$.
\end{enumerate}
The process ends when step 2~picks a vertex $x \in M$.
\end{quote}

Denoting by $\{Y_0=s,Y_1,\dots,Y_\rho \in M\}$ the Markov chain corresponding to the sources picked in the elfs process, we define the electric hitting time $\EHT_s = \mathbb{E}[\rho]$ as the expected time before the elfs hit the sink and the process terminates. 
Somewhat ethereally, it is shown in \cite{apers2022elfs} that the final sink vertex follows the arrival distribution of a random walk $\{X_0=s,X_1,\dots,X_\tau \in M\}$ from source $s$ to sink $M$.
Even stronger, we can couple the elfs process and the random walk by introducing random stopping times $0 < \nu_1 < \dots < \nu_\rho = \tau$ such that $Y_i = X_{\nu_i}$, as is illustrated in \zcref{fig:intro-coupling}.
The expected number of random walk steps between the elfs turns out to be given by the aforementioned escape time,
\[
\mathbb{E}[\nu_{i+1} - \nu_i] = \ET_{Y_i}/2,
\]
from which we get that
\[
\mathbb{E}\left[\sum_{i=0}^{\rho-1} \ET_{Y_i}\right]
= 2 \HT_s.
\]
This clearly implies that $\EHT_s \leq 2\HT_s$ and this difference can be exponential, e.g.~on an $n$-vertex path graph, in which case $\HT_s \in \Theta(n^2)$ while $\EHT_s \in \Theta(\log n)$.

\begin{figure}[htb]
\centering
\includegraphics[width=.6\textwidth]{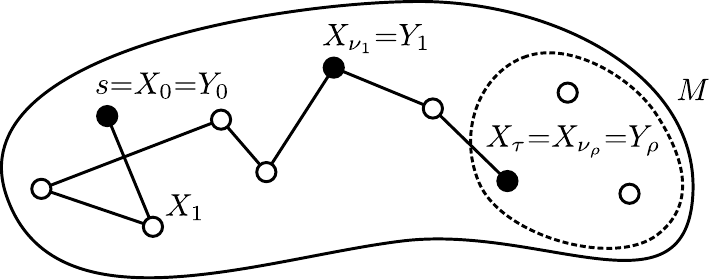}
\caption{Coupling between a random walk and elfs process through stopping rules $\nu_1<\dots<\nu_\rho$. Figure from \cite{apers2022elfs}.}
\label{fig:intro-coupling}
\end{figure}

The interest in this process is that, by using quantum walks to prepare elfs, it yields a way of quantumly sampling from the random walk arrival distribution, and with it find a quantum speedup for its many applications \cite{broder1989generating,zhu2003semi,chung2013solving,linden2022quantum}.
In an ideal implementation we would charge $\sqrt{\ET_x}$ for an electric flow sample from $x$, in which case the expected complexity would be
\[
\mathbb{E}\left[\sum_{i=0}^{\rho-1} \sqrt{\ET_{Y_i}}\right]
\in O(\sqrt{\EHT_s \HT_s}),
\]
where the bound follows from Jensen's inequality.
This can yield up to a quadratic speedup over a random walk, and it was the main motivation for studying this process in \cite{apers2022elfs}.
Crucially, however, the elfs process may require composing up to $\poly(|V|)$ many elfs, and in that case the overhead from using the existing low-precision algorithms washes away the quantum speedup.
In particular, the complexity in \cite{apers2022elfs} is $\tO\big(\sqrt{\EHT_s^3 \HT_s}\big)$, and this can be significantly bigger than~$\HT_s$.

Leaving it as maybe the main open question in \cite{apers2022elfs}, it was conjectured that there exists a more efficient way of composing elfs.
We resolve this question in the positive by using our elfs transducer and exploiting the elegant composition properties of transducers.

\begin{theorem}[Informal] \label{thm:intro-elfs-process}
Given upper bounds $\bar\EHT_s \geq \EHT_s$ and $\bar\HT_s \geq \HT_s$, there exists an (infinite-dimensional) transducer $\tilde V$, making one call to the quantum walk operator, that samples from the random walk arrival distribution with transduction complexity
\[
\tO\left(\sqrt{\bar\EHT_s \bar\HT_s}\right).
\]
\end{theorem}

This follows from \zcref{thm:elfs-process-global} with the parameter choice
in \zcref{lem:para-choice}.
For the special case of quantum walk search, where we are only interested in returning an element from $M$ (and we do not care about the distribution of the returned element), we show that no upper bounds are required.
There is a sweeping strategy that returns an element from $M$ while taking $\tO(\sqrt{\EHT_s \HT_s})$ many quantum walk steps.

\subsection{Result 5: Elfs and learning on expanders} \label{sec:elfs-learning}

We show that on expander graphs we can get uniform bounds on all relevant quantities as a function of the number of vertices $n$ and the size of the sink $m = |M|$ (and this allows us to avoid the log-factors in \zcref{thm:intro-elfs-process}).
Turning the elfs process transducer into a (finite-dimensional) algorithm then yields a quantum algorithm that samples from a distribution $\varepsilon$-close to the random walk arrival distribution in
\[
O(\sqrt{n}/\varepsilon^2)
\]
quantum walk steps.
When $m \leq \sqrt{n}$ this is an up-to-quadratic speedup over a random walk, whose hitting time $\HT_s$ scales as $n/m$.
The formal sampling guarantee is given in \zcref{thm:elfs-expander}.

Demonstrating the utility of this quantum algorithm, we can use it to implement the semi-supervised learning algorithm from Zhu, Ghahramani and Lafferty \cite{zhu2003semi} on expander graphs.
In its simplest form, $n$ data points are represented by the vertices $V$ of a graph $G$, with the edges~$E$ capturing similarity between the data points.
A subset $M \subseteq V$ of the data points (typically with $m = |M| \ll n$) is labeled, with binary labels $\{b_x\}_{x \in M}$.
The task is to assign a label to an unlabeled vertex $s \in V \backslash M$.
The idea in \cite{zhu2003semi} is to assign to $s$ the number
\[
\tilde b_s
= \sum_{x \in M} p_x b_x
\]
with $p$ the arrival distribution of a random walk from $s$ to $M$, and to then round this number to obtain a label for $s$.


For constant additive accuracy and success probability at least $2/3$, estimating $\tilde b_s$ requires only a constant number of label samples from a distribution sufficiently close to $p$ in total variation distance. Assuming unit-cost label queries, this takes $O(\HT_s)=O(n/m+\log n)$ expected random walk steps classically, or $O(\sqrt n)$ quantum walk steps using our algorithm.
If $|\tilde b_s-1/2|$ is bounded below by a positive constant, choosing a sufficiently small constant estimation error ensures that rounding gives the same label as rounding $\tilde b_s$.
More generally, assuming unit-cost coherent label queries and access to the sampling circuit and its inverse, quantum amplitude estimation yields an estimate of $\tilde b_s$ to additive error $\delta$ with probability at least $2/3$, using $O(\sqrt n/\delta^3)$ quantum walk steps, by setting the sampling error to $\delta/2$.

For comparison, while existing quantum walk search algorithms \cite{belovs2013quantum,piddock2019quantum,apers2019unified} can find an element in $M$ from $s$ in $O(\sqrt{\CT_s}) \in O(\sqrt{n})$, these algorithms have no control over the distribution of the returned element.

\subsection{Discussion and open questions}

We used the recent theory of transducers \cite{belovs2024taming} to push forward the algorithmic theory of quantum walks and electric networks.
In particular, we showed that there exist zero-error transducers for reflecting around and preparing elfs.
Composing these leads to strictly improved quantum walk algorithms for estimating effective resistances and for solving sampling and learning tasks on graphs.
As the key technical tool, we established a zero-error transducer version of the effective gap lemma.
We anticipate this to have broader implications --- e.g., we use it for improved witness size estimation of span programs.
Towards extending our work, we leave the following for future work:
\begin{itemize}
\item
The elfs process transducer (\zcref{thm:intro-elfs-process}) effectively simulates deferred measurements of the elfs, which implies that the space complexity scales with the number of elfs composed.
This contrasts with the approach in \cite{apers2022elfs}, where the space complexity is essentially that of the quantum walk operator (though the quantum speedup is smaller).
Is there a space-efficient way to simulate measurements in the transducer framework?
\item
Aside from learning on graphs, further applications of sampling from the random walk arrival distribution include solving heat equations \cite{chung2013solving,linden2022quantum} and sampling random spanning trees \cite{broder1989generating}.
Can we use our elfs process transducer to find a quantum speedup for these?
\item
Finally, as we discussed, the elfs process can speed up search on graphs.
A long line of work \cite{szegedy2004quantum,magniez2007search,krovi2016quantum,ambainis2020quadratic,piddock2019quantum} established a quadratic speedup over the random walk hitting time when starting from a quantum version of the stationary distribution.
It remains a very interesting open question whether quantum walks can generally speed up the hitting time from a single vertex.
In~\cite{apers2022elfs} this question was resolved on trees using the elfs process.
We hope that our elfs process transducer can help to resolve the question on other graphs as well.
\end{itemize}

\section{Transducers and the effective gap lemma}

In this section we introduce transducers and prove our error-free transducer version of the effective gap lemma (\zcref{thm:eff-gap-transducer}).

\subsection{Transducers}

We use the following definition from \cite{belovs2024taming}.

\begin{definition}
    Let $S$ be a unitary that acts on a direct sum of two vector spaces $\mathcal{H} \oplus \mathcal{L}$.
    For every $\xi \in \mathcal{H}$, there exist $\tau=\tau(S, \xi) \in \mathcal{H}$ and $w=w(S, \xi) \in \mathcal{L}$ such that
    \[
    S: \xi \oplus w \mapsto \tau \oplus w.
    \]
    Moreover, the vector $\tau \in \mathcal{H}$ is uniquely determined by $\xi$ and $S$, and the mapping $\xi \mapsto \tau$ is unitary.
\end{definition}

We say that $S$ transduces $\xi$ into $\tau$, and denote this by $\xi \stackrel{S}{\leadsto} \tau$.
The vector $w$ is called the catalyst of this transduction, and $W(S, \xi)=\|w(S, \xi)\|^2$ is the transduction complexity of $S$ on $\xi$.
We call~$\mc H$ the \emph{public space} of $S$ and $\mc L$ the \emph{private space}.
The transducer defines a mapping $\xi \to \tau$ for every $\xi \in \mc H$ and it turns out that this mapping is unitary \cite[Theorem 5.1(c)]{belovs2024taming}.
We call it the \emph{transduction action} and denote it by $S\DownTransduce_{\mc H}$.

The notion of a transducer is motivated by the following lemma, which shows that a transducer can be turned into a bounded-error quantum algorithm.

\begin{lemma}[Theorem 5.5 of \cite{belovs2024taming}] \label{lem:trans-to-alg}
    Let $K\in\mb N$.
    For every transducer $S: \mathcal{H} \oplus \mathcal{L} \rightarrow \mathcal{H} \oplus \mathcal{L}$ and initial state $\ket{\xi} \in \mc H$, there exists a quantum algorithm that transforms $\ket{\xi}$ into $\ket{\tau^{\prime}}$ such that
    \[
    \left\|\ket{\tau^{\prime}}-\ket{\tau(S, \xi)}\right\| \leqslant 2 \sqrt{\frac{W(S, \xi)}{K}}.
    \]
    The algorithm uses $K$ controlled applications of $S$, $O(K)$ controlled applications of $2P_{\mathcal L}-I$ and $O(K)$ other elementary operations, where $P_{\mathcal L}$ is the projector onto the private space $\mathcal L$.\footnote{When the public and private spaces are distinguished by a single qubit, as in~\cite{belovs2024taming}, each controlled reflection $2P_{\mathcal L}-I$ requires only $O(1)$ elementary operations.}
\end{lemma}

We will repeatedly invoke transducer composition.
In the most general case that we will require, we consider a potentially infinite family of transducers $S_0, S_1, \dots, S_{m-1}$ (with $m \in \mathbb{N} \cup \{\infty\}$) on the same public space $\mathcal{H}$, with $\mathcal{H} = \mathcal{H}_0 \oplus \mathcal{H}_1$.
We are interested in the sequence of transductions
\begin{align*}
	\begin{aligned}
		&\psi_{0,0}\;\overset{S_0}{\rightsquigarrow}\;\\[-7pt]
		&\;\\[-7pt]
		&\;
	\end{aligned}
	\left\{
	\begin{aligned}
	&\psi_{1,0}\;\overset{S_1}{\rightsquigarrow}\;\\[-7pt]
	&\;{\scriptscriptstyle+}\\[-7pt]
	&\psi_{1,1}
\end{aligned}
	\right.
	\left\{
\begin{aligned}
	&\psi_{2,0}\;\overset{S_2}{\rightsquigarrow}\;\\[-7pt]
	&\;{\scriptscriptstyle+}\\[-7pt]
	&\psi_{2,1}
\end{aligned}
\right.
\begin{aligned}
	&\ldots\;\\[-7pt]
	&\;\\[-7pt]
	&\;
\end{aligned}
\end{align*}
where $\psi_{t,i} \in \mathcal{H}_i$ and $\psi_{t,0} \overset{S_t}{\rightsquigarrow} \psi_{t+1,0} + \psi_{t+1,1}$ with catalyst $w_t$ and transduction complexity $W(S_t,\psi_{t,0}) \eqcolon \|\psi_{t,0}\|^2 \cdot W_t$.
We can compose these into one global transducer using the following lemma.
It is a variation of \cite[Proposition 9.9]{belovs2024taming}, which covers the case when $m$ is finite and $\psi_{t,1} = 0$ for all $t$.
We use the notation $T_{c,m}$ to denote the cost of increasing and conditioning operations on a length-$(m+1)$ counter.
If $m$ is finite then $T_{c,m} \in O(\log m)$.

\begin{lemma}[Transducer composition] \label{lem:contr-transd-comp}
Let $m \in \mathbb{N} \cup \{\infty\}$.
There exists a transducer $S$ with a counter of length $m+1$ that transduces $\psi_{0,0}$ into the final state $\tau = \sum_{t=1}^m \ket{t} \otimes \psi_{t,1}$, with transduction complexity and time complexity\footnote{We write $T(S)$ for the usual time complexity of implementing $S$ as a unitary.}
    \[
    W(S, \xi)=W_0+ \sum_{t=1}^{m-1} \|\psi_{t,0}\|^2 (1+W_t), \qquad
    T(S) = T_{c,m} + \max_{0 \leq t < m} T(S_t).
    \]
\end{lemma}
\begin{proof}
Let $V_c$ be the operation that acts as $\sum_{t}\ket{t+1}\bra{t}$ (increasing the counter) on the first register conditionally on the state of the second register lying in the public space $\mathcal{H}$ of the $S_t$'s, and as the identity otherwise.
Then the unitary
\[
S
= V_c \cdot \sum_{t=0}^{m-1} \ket{t}\bra{t} \otimes S_t
\]
transduces $\ket{0} \otimes \psi_{0,0} \overset{S}{\rightsquigarrow} \sum_{t=1}^m \ket{t} \otimes \psi_{t,1}$ with catalyst
$
w = \ket{0}\otimes w_0 + \sum_{t=1}^{m-1} \ket{t} \otimes (\psi_{t,0}+w_t)
$
and the claimed transduction complexity.
The time complexity of $S$ is determined by the cost $T_{c,m}$ of increasing and conditioning on the counter, and the cost $\max_{0 \leq t < m} T(S_t)$ of doing a single, parallel controlled call to the transducers $S_t$.
\end{proof}

\begin{remark}[Infinite dimensions]
Similar to some constructions in \cite{belovs2024taming,belovs2025space,jeffery2025qma}, the resulting transducers might use an infinite counter as a control, and (unlike \cite{belovs2024taming,belovs2025space,jeffery2025qma}) we will even construct a transducer that uses an infinite array of qubits.
The key point is that the actual quantum algorithm derived from the transducer (\zcref{lem:trans-to-alg}) only makes $K$ calls to the transducer $S$, and that $S^k$ applied to the initial state $\psi_{0,0}$ will only touch on the first $\tO(K)$ states of the counter and the first $\tO(K)$ many qubits of the array.
We can thus truncate these spaces, and the resulting quantum algorithm will be finite-dimensional.
See \cite[Section 4.1]{belovs2025space} for a more detailed discussion.
\end{remark}

\subsection{Effective gap transducer}
\label{subsec:effective-gap-trans}

Here we restate the theorem from the introduction:
\effgap*

Equivalently, we have that
\[
U\DownTransduce_{\ker\Pi}
= 2P_{\ker(\Pi) \cap \ker(\Delta)} - I.
\]
The theorem essentially follows from the following lemma.

\begin{lemma} \label{lem:proj}
If $\ket{\psi} \in \ker(\Pi)$ and $w = (\Pi-\Pi \Delta \Pi)^+ \Delta \ket{\psi}$ then $(I-\Delta) (\ket{\psi}+w) = P_{\ker(\Pi) \cap \ker(\Delta)} \ket{\psi}$.
\end{lemma}
\begin{proof}
By \cite[Eq.(4.7)]{piziak1999constructing}, we have
\[
    P_{\ker(\Pi) \cap \ker(\Delta)} = (I-\Delta) - (\Pi(I-\Delta))^+ \Pi (I-\Delta).
\]
The Moore--Penrose identity $A^+=A^\dagger(AA^\dagger)^+$ implies
\[
(\Pi(I-\Delta))^+ = (I-\Delta) \Pi (\Pi - \Pi \Delta \Pi)^+.
\]
Since
$\Pi(\Pi-\Pi\Delta\Pi)^+\Pi=(\Pi-\Pi\Delta\Pi)^+$,
we obtain
\[
P_{\ker(\Pi) \cap \ker(\Delta)}
= (I-\Delta)\left[ I - (\Pi - \Pi \Delta \Pi)^+ (I-\Delta) \right].
\]
We can therefore rewrite
\[
P_{\ker(\Pi) \cap \ker(\Delta)} \ket{\psi}
= (I-\Delta)\left[ I - (\Pi - \Pi \Delta \Pi)^+ (I-\Delta) \right] \ket{\psi}
= (I-\Delta)\left[ I + (\Pi - \Pi \Delta \Pi)^+ \Delta \right] \ket{\psi}
\]
where the last step uses that $(\Pi - \Pi \Delta \Pi)^+ \ket{\psi} = 0$ since $\ket{\psi} \in \ker(\Pi)$.
\end{proof}

Using it, we can prove a slight generalization of \zcref{thm:eff-gap-transducer}.
The generalization shows how to do also partial rotations around the intersection of kernels, as we will require later.

\begin{theorem} \label{thm:eff-gap-transducer-general}
Consider projectors $\Pi$ and $\Delta$, and define the partial rotation $U(\theta) = \Delta+e^{i \theta}(I-\Delta)$ over angle $\theta$.
Let \(P=P_{\ker(\Pi) \cap \ker(\Delta)}\). For any $\ket{\psi} \in \ker(\Pi)$, it holds that
\[
\ket{\psi}
\overset{U(\theta)}{\rightsquigarrow} \bigl((I-P)+e^{i\theta}P\bigr)
\ket{\psi}
\]
with public space $\ker(\Pi)$ and catalyst $w = (\Pi-\Pi \Delta \Pi)^+ \Delta \ket{\psi} \in \mathrm{im}(\Pi)$.
\end{theorem}
\begin{proof}
From \zcref{lem:proj} we get that
\begin{align*}
U(\theta) \left(\ket{\psi} \oplus w \right)
&= \big( \Delta+e^{i \theta}(I-\Delta) \big) \left(\ket{\psi} \oplus w \right)\\
&= \ket{\psi} \oplus w + (e^{i \theta}-1)(I-\Delta)
(\ket{\psi}\oplus w) \\
&= \bigl((I-P)+e^{i\theta}P\bigr) \ket{\psi} \oplus w. \qedhere
\end{align*}
\end{proof}

\zcref{thm:eff-gap-transducer} then follows from setting $\theta = \pi$, and adding the extra reflection $2\Pi-I$ which only flips the sign of the transduction.
We note that our expression for the catalyst corresponds to the generic expression of the catalyst given in \cite[Eq.~(5.2)]{belovs2024taming}, which is
\[
w
= (\Pi - \Pi U \Pi)^+ \Pi U \ket{\psi}
= (\Pi - \Pi U \Pi)^+ U \ket{\psi}.
\]
For the special case where $U = (2\Pi-I)(2\Delta-I)$, and using that $\ket{\psi} \in \ker(\Pi)$, this simplifies to
\begin{align*}
w
&= (\Pi - \Pi U \Pi)^+ \Pi U \ket{\psi} \\
&= (\Pi - \Pi (2\Delta-I) \Pi)^+ (2\Delta-I) \ket{\psi} \\
&= 2 (\Pi - \Pi (2\Delta-I) \Pi)^+ \Delta \ket{\psi}
= (\Pi - \Pi \Delta \Pi)^+ \Delta \ket{\psi}.
\end{align*}

\section{Elfs and quantum walks}

\subsection{Graphs, electric networks and random walks}

Here we introduce key concepts from graphs, electric networks and random walks, making some concepts from the introduction more precise.
We refer the interested reader to \cite{doyle1984random,lyons2017probability} for more details.

We consider an undirected graph $G=(V, E, w)$ with edge set $E \subseteq V \times V$ so that $(x,y) \in E$ if and only if $(y,x) \in E$, and with nonnegative weights $w: E \rightarrow \mathbb{R}_{\geq 0}$ satisfying $w_{xy}=w_{yx}$ for all $(x,y)\in E$.
For convenience, for $(x,y) \notin E$ we set $w_{xy}=0$.
We let $d_x = \sum_y w_{xy}$ denote the degree at $x$, and $W = \sum_{(x,y) \in E} w_{xy} = \sum_x d_x$ denote twice the total edge weight.

For source $s \in V$ and sink $M \subseteq V$, a unit $s$-$M$ flow $g:E \to \mathbb{R}$ is a function on the edges that is antisymmetric, $g_{xy} = -g_{yx}$, and satisfies the demands
\[
\sum_y g_{sy} = 1, \qquad
\sum_{m \in M} \sum_y g_{ym} = 1, \qquad
\sum_z g_{xz} = 0, \;\; \forall x \notin \{s\} \cup M.
\]
The \emph{unit electric flow} from $s$ to $M$ is the unique unit $s$-$M$ flow $f: E \rightarrow \mathbb{R}$ that minimizes the dissipated energy
\[
\mathcal{E}(f)
= \frac{1}{2} \sum_{(x,y) \in E} \frac{f_{xy}^2}{w_{xy}} .
\]
The resulting energy of the electric flow is called the effective resistance $R_s$.
Equivalently, the electric flow is also the unique unit $s$-$M$ flow that is a potential flow: there exist voltages $v:V \to \mathbb{R}_{\geq 0}$ such that
\begin{equation} \label{eq:f-volt}
f_{x y}=w_{x y}\left(v_x-v_y\right), \quad \forall (x,y) \in E.
\end{equation}
As a convention, we set $v_y=0$ for all $y \in M$, in which case $v_s=R_s$.

The electric flow is tightly connected to properties of a simple random walk on $G$, which is described by a Markov chain $\{X_0,X_1,\dots\}$ on $V$ with $\mathbb{P}(X_{t+1} = y \mid X_t = x) = w_{xy}/d_x$.
For a random walk starting from $s$, so $X_0 = s$, we will be interested in the following quantities:
\begin{itemize}
\item
The \emph{escape probability} $p_s$ is the probability that the random walk from $s$ hits $M$ before returning to $s$.
\item
The \emph{hitting time} $\HT_s = \mathbb{E}[\tau_M]$ is the expectation of $\tau_M = \min\{t : X_t \in M\}$, which is the number of steps until the random walk hits $M$.
\item
The \emph{escape time} $\ET_s = \mathbb{E}[\sigma]$ is the expectation of $\sigma = 1 + \max\{t < \tau_M : X_t = s\}$, which is the time at which the random walk left $s$ for the final time before hitting $M$.
\item
The \emph{commute time} $\CT_s = \mathbb{E}[\kappa]$ is the expectation of $\kappa = \min\{t > \tau_M : X_t = s\}$, which is the first return time to $s$ after hitting $M$.
\end{itemize}
We can express these quantities using the electric potential:
\[
\frac{1}{p_s}
= R_s d_s
\quad \leq \quad
\ET_s
= \frac{1}{R_s} \sum_x v_x^2 d_x
\quad \leq \quad
\HT_s
= \sum_x v_x d_x
\quad \leq \quad
\CT_s
= R_s \sum_x d_x
= R_s W
\]
where the inequalities follow from the fact that $v_x \leq v_s=R_s$ for all $x$.
While the expressions for the hitting time and the commute time have long been known \cite{doyle1984random,chandra1989electrical}, the expression for the escape time was proven in \cite{apers2022elfs}.

\subsection{Quantum walks and elfs}

The edge space of $G$, denoted $\ell(E)$, is defined as the Hilbert space with the standard orthonormal basis $\{|x y\rangle,|y x\rangle : (x, y) \in E\}$.
The symmetric and antisymmetric subspaces $\ell_{\pm}(E) \subset \ell(E)$ are defined as
\[
\ell_\pm(E)=\on{span}\left\{\ket{xy} \pm \ket{yx} : (x,y) \in E\right\}.
\]
The star state $\left|\phi_x\right\rangle$ at $x$ and the star subspace $\ell_*(E) \subset \ell(E)$ are defined as 
\[
\left|\phi_x\right\rangle=\frac{1}{\sqrt{d_x}} \sum_y \sqrt{w_{x y}}|x y\rangle, 
\qquad 
\ell_*(E)=\on{span}\{\left|\phi_x\right\rangle : x\in V \backslash (\{s\}\cup M)\}.
\]
With $f$ the unit electric $s$-$M$ flow, we define the electric flow state $\ket{f}$ as the normalized vector
\begin{equation}
\label{eq:elf state}
    \ket{f} = \frac{1}{\sqrt{2 R_s}} \sum_{(x,y) \in E} \frac{f_{xy}}{\sqrt{w_{xy}}}\ket{xy}.
\end{equation}
The state is antisymmetric, i.e. $\ket{f} \in \ell_-(E)$, and satisfies
\[
\sqrt{d_s} \inner{\phi_s|f} = \frac{1}{\sqrt{2R_s}}, \qquad
\sum_{m \in M} \sqrt{d_m} \inner{\phi_m|f} = -\frac{1}{\sqrt{2R_s}}, \qquad
\inner{\phi_x|f} = 0, \;\; \forall x \notin \{s\} \cup M.
\]

Following \cite{apers2022elfs}, we define a quantum walk on $G$ as the unitary operator $U$ on $\ell(E)$ given by
\begin{equation} \label{eq:QW-operator}
U
= (2\Pi_*-I)\cdot \mathrm{SWAP}
= (2\Pi_*-I)\cdot (2\Pi_+ - I)
\end{equation}
with $\mathrm{SWAP} \ket{x,y} = \ket{y,x}$, $\Pi_+$ the projector onto the symmetric subspace $\ell_+(E)$, and $\Pi_*$ the projector onto the star subspace $\ell_*(E)$.
The following is a useful fact, where a flow $g$ is ``closed outside $M$'' if it is antisymmetric and $\inner{\phi_x|g} = 0$ for all $x \notin M$ (including $s$).
\begin{lemma}[Lemma 9 of \cite{apers2022elfs}]
The invariant space of $U$ equals $\ker(\Pi_+) \cap \ker(\Pi_*)$, and it is spanned by the electric flow $\ket{f}$ and all flows $\ket{g}$ that are closed outside $M$.
\end{lemma}

Now let $\Pi_{\mathrm{cl}}$ denote the projector onto the flows closed outside $M$.
Then, by Kirchhoff's cycle law \cite{lyons2017probability}, we have that $\Pi_{\mathrm{cl}} \ket{f} = 0$ and so we can decompose
\[
P_{\ker(\Pi_+) \cap \ker(\Pi_*)}
= \ket{f}\bra{f} + \Pi_{\mathrm{cl}}.
\]
Since $\Pi_{\mathrm{cl}} \ket{\phi_s} = 0$ by definition, we get the useful consequence that
\begin{equation} \label{eq:elf-proj}
P_{\ker(\Pi_+) \cap \ker(\Pi_*)} \ket{\phi_s}
= \inner{f|\phi_s} \ket{f}
= \frac{1}{\sqrt{2 R_s d_s}} \ket{f}.
\end{equation}

\subsection{Elfs transducer}
Having established \zcref{eq:elf-proj}, we can use a transducer to reflect about $\ket{f}$.
Specifically, we prove the following theorem, which is a special case of the effective gap transducer (\zcref{thm:eff-gap-transducer}).
\begin{theorem}[Elfs transducer] \label{thm:elf-transd-refl}
With $U = (2\Pi_*-I) (2\Pi_+ - I)$ we have that
\[
\ket{\phi_s}
\overset{U}{\rightsquigarrow} (2\ket{f}\bra{f}-I) \ket{\phi_s}
\]
with public space $\ker(\Pi_*)$, catalyst
\[
w
= \frac{1}{R_s \sqrt{d_s}} \sum_{x\neq s} v_x \sqrt{d_x} \ket{\phi_x}
\in \mathrm{Im}(\Pi_*)
\]
and transduction complexity $W = \frac{\ET_s}{R_s d_s} - 1$.
\end{theorem}
\begin{proof}
From \zcref{eq:f-volt} it follows that
\[
(I-\Pi_+) (\ket{\phi_s} \oplus w)
= \frac{1}{\sqrt{2 R_s d_s}} \ket{f}
= P_{\ker(\Pi_+) \cap \ker(\Pi_*)} \ket{\phi_s}.
\]
Following \zcref{eq:elf-proj}, this implies that
\[
U \left(\ket{\phi_s} \oplus w \right)
= (2\ket{f}\bra{f}-I)\ket{\phi_s} \oplus w. \qedhere
\]
\end{proof}

Since we also have that
\[
\ket{f} \overset{U}{\rightsquigarrow} \ket{f}
\]
with the zero catalyst, this fully characterizes the transduction action of $U$ on the subspace spanned by $\ket{f}$ and $\ket{\phi_s}$.

\subsection{Graph modification} \label{sec:graph-mod}

As will be required for our later applications, and in analogy to earlier works \cite{belovs2013quantum,ito2019approximate,piddock2019quantum,apers2019unified,apers2022elfs}, we need to rebalance terms to get good algorithms.
In particular, we will need to minimize the product of the elfs transduction complexity, $\frac{\ET_s}{R_s d_s}-1$ (\zcref{thm:elf-transd-refl}), with the inverse of the overlap, $\frac{1}{\inner{f|\phi_s}} = \sqrt{2 R_s d_s}$, yielding a target complexity scaling as $\frac{\ET_s}{\sqrt{R_s d_s}}$.
To minimize this complexity we create a modified graph $\hat G=(V\cup\{\sigma\},E\cup\{(\sigma,s)\}, \hat w)$ by attaching a new vertex~$\sigma$ to the source $s$ via an additional edge $(\sigma,s)$ of weight $w_{\sigma s} = \eta d_s$ (see \zcref{fig:R-graph}).
As discussed in \cite[Lemma 19]{apers2019unified}, we can easily modify the quantum walk operator on $G$ to implement the quantum walk operator on $\hat G$ (even without knowing $d_s$).

\begin{figure}[H]
\centering
\includegraphics[width=.4\textwidth]{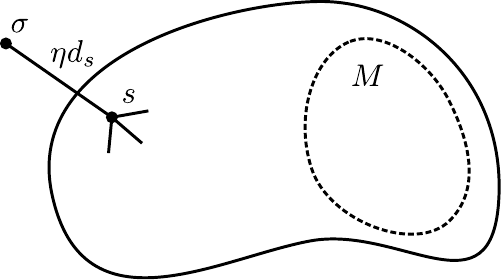}
\caption{Modified graph $\hat G$. Figure derived from \cite{apers2022elfs}.}
\label{fig:R-graph}
\end{figure}

The voltages, resistance and degrees on $\hat G$ are
\[
{\hat v}_\sigma = {\hat R}_\sigma = R_s + \frac{1}{\eta d_s}, 
\qquad {\hat v}_s = v_s = R_s,
\qquad {\hat d}_\sigma = \eta d_s, 
\qquad {\hat d}_s = (1+\eta)d_s,
\]
and ${\hat v}_x = v_x$ and ${\hat d}_x = d_x$ for all $x \notin \{\sigma,s\}$.
We can prove the following bound on the escape time.

\begin{lemma} \label{lem:ET-mod-graph}
For $\eta >0$, we have 
\[
\hat{R}_\sigma \hat{d}_\sigma = 1 + \eta R_s d_s
\quad \text{ and } \quad
\frac{\hat\ET_\sigma}{\hat{R}_\sigma\hat{d}_\sigma}
\leq 2 + \frac{\ET_s}{\eta R_s d_s}.
\]
\end{lemma}
\begin{proof}
The equality is direct.
For the inequality, we first note that
\[
    \hat\ET_\sigma
    =\frac{1}{\hat R_\sigma} \sum_x {\hat v}_x^2 {\hat d}_x
    = \hat{R}_\sigma \hat{d}_\sigma + \frac{\hat{v}_s^2 \hat{d}_s}{\hat{R}_\sigma} + \frac{1}{\hat R_\sigma} \sum_{x\neq s,\sigma} v_x^2 d_x
= \hat{R}_\sigma \hat{d}_\sigma + \eta \frac{R_s^2}{\hat{R}_\sigma} d_s + \frac{R_s}{\hat R_\sigma} \ET_s.
\]
Then we use $\hat{R}_\sigma \geq R_s$ to bound
\[
\frac{\hat\ET_\sigma}{\hat{R}_\sigma\hat{d}_\sigma}
= 1 +  \eta\frac{R_s^2 d_s}{\hat{R}_\sigma^2 \hat{d}_\sigma} + \frac{R_s}{\hat{R}_\sigma^2 \hat{d}_\sigma} \ET_s
\leq 2 + \frac{\ET_s}{\eta R_s d_s}. \qedhere
\]
\end{proof}

This yields a bound on the target complexity of the modified graph
\[
\frac{\hat{\ET}_\sigma}{\sqrt{\hat{R}_\sigma \hat{d}_\sigma}}
\leq \left( 2 + \frac{\ET_s}{\eta R_s d_s} \right) \sqrt{1 + \eta R_s d_s}.
\]
By picking $\eta = \frac{\bar\ET_s}{R_s d_s}$ for $\bar\ET_s \geq \ET_s$, this is $O(\sqrt{\bar\ET_s})$.

\section{Effective resistance estimation}

Through a first non-trivial composition of our elfs transducer with quantum phase estimation, we obtain a quantum algorithm for estimating the escape probability or effective resistance in a graph.
For this, we slightly modify the quantum walk by setting
\[
U'
= (I-2\ket{\phi_s}\bra{\phi_s}) U
= \left( 2 \sum_{x \notin M} \ket{\phi_x}\bra{\phi_x}-I \right) (2\Pi_+-I).
\]
Then the transduction action of $U'$ on $\ket{\psi} \in \mathrm{span}\{\ket{\phi_s},\ket{f}\}$ (with the same public space $\ker\Pi_*$ and catalyst $w$) is $\ket{\psi} \overset{U'}{\rightsquigarrow} R \ket{\psi}$, with $R$ the rotation
\begin{equation} \label{eq:rotation}
R
= (I-2\ket{\phi_s}\bra{\phi_s}) (2\ket{f}\bra{f}-I).
\end{equation}
This rotation acts non-trivially only on $\mathrm{span}\{\ket{\phi_s},\ket{f}\}$.
In that subspace, $R$ has eigenvectors $\ket{v^\pm}$ with eigenvalues $e^{\pm i 2 \theta}$ for $\sin\theta = |\inner{f|\phi_s}| = \frac{1}{\sqrt{2 R_s d_s}}$, and we can rewrite $\ket{\phi_s} = \frac{e^{i\theta}}{\sqrt{2}} \ket{v^+} + \frac{e^{-i\theta}}{\sqrt{2}} \ket{v^-}$ \cite{brassard2000quantum}.

We can use quantum phase estimation on $\ket{\phi_s}$ to get an estimate of $\theta$.
This requires access to a controlled implementation of powers of $R$, which we can easily obtain by composing the transducer for $R$.

\begin{lemma} \label{lem:contr-transd}
Given integer $T \geq 1$ and upper bound $\eta \geq \frac{\ET_s}{R_s d_s}$, there is a quantum algorithm that maps $\ket{\phi_s}\ket{0}$ to a state $9/10$-close to
\[
\frac{1}{\sqrt{T}} \sum_{t=0}^{T-1} R^t \ket{\phi_s} \ket{t}
\]
while making $O(\eta T)$ calls to $U'$ and $\widetilde{O}(\eta T)$ other elementary operations.
\end{lemma}
\begin{proof}
We can implement the transduction $\ket{\phi_s}\ket{0} \rightsquigarrow \frac{1}{\sqrt{T}} \sum_{t=0}^{T-1} R^t \ket{\phi_s} \ket{t}$ by using transducer composition (\zcref{lem:contr-transd-comp}) with $S_t = U' \otimes \Pi_{> t} + I \otimes \Pi_{\leq t}$.
The transduction complexity is bounded by $\frac{\ET_s}{R_s d_s} T$ and the transducer effectively makes 1 call to the controlled quantum walk operator $U'$ and $O(\log T)$ other elementary operations.

The lemma then follows from invoking \zcref{lem:trans-to-alg} to turn this transducer into a quantum algorithm.
With $K \in \Theta(\eta T)$ this yields a quantum algorithm that returns a state that is $9/10$-close to the transduced state while making $O(K)$ calls to the transducer and other elementary operations.
\end{proof}

By standard results \cite[Appendix C]{cleve1998quantum}, if we apply the quantum Fourier transform on the $\log(T)$ ancilla qubits of a state that is $9/10$-close to $\frac{1}{\sqrt{T}} \sum_{t=0}^{T-1} R^t \ket{\phi_s} \ket{t}$, then measuring these qubits returns with probability at least $2/3$ an estimate $\tilde\theta = \theta + O(1/T)$ (and hence $\sin(\tilde\theta) = 1/\sqrt{2R_sd_s} + O(1/T)$).
This implies the following lemma.

\begin{lemma} \label{lem:eff-res-prim}
Given an upper bound $\tau \geq \frac{\ET_s}{R_s d_s}$ and integer $T \geq 1$, there is a bounded-error quantum algorithm that returns an estimate $\frac{1}{\sqrt{R_s d_s}} + O(1/T)$ while making $O(\tau T)$ quantum walk steps and $\widetilde{O}(\tau T)$ other elementary operations.
\end{lemma}

\subsection{Known estimate} \label{sec:eff-res-known}
Given an upper bound $\bar\ET_s \geq \ET_s$, and assuming that we already have a constant-factor estimate $R_s d_s/2 \leq p \leq R_s d_s$, we can set $\tau = \frac{\bar\ET_s}{p} \geq \frac{\ET_s}{R_s d_s}$ and $T \in \Theta(\frac{1}{\varepsilon \sqrt{p}})$ to obtain a multiplicative estimate $(1 + O(\varepsilon)) R_s d_s$ using a number of quantum walk steps
\[
O\left( \frac{1}{\varepsilon} \frac{\bar\ET_s}{\sqrt{R_s d_s}} \right).
\]
Applying this algorithm to the modified graph $\hat G$ from \zcref{sec:graph-mod} with $\eta = \tau$ yields
\[
O\left( \frac{1}{\varepsilon} \frac{\bar\ET_s}{\sqrt{\hat R_\sigma \hat d_\sigma}} \right)
\overset{\text{\zcref{lem:ET-mod-graph}}}{\in} O\left(\frac{1}{\varepsilon} \sqrt{\bar\ET_s}\right)
\]
quantum walk steps and $\widetilde{O}\left(\frac{1}{\varepsilon} \sqrt{\bar\ET_s}\right)$ other elementary operations for returning an estimate $\hat p = (1 + O(\varepsilon)) \hat R_\sigma \hat d_\sigma$.
We can then turn this into an estimate $(\hat p-1)/\eta = (1 + O(\varepsilon)) R_s d_s$.

\subsection{Unknown estimate}
\label{subsec:eff-res-unknown}

In the case where we only have an upper bound $\bar\ET_s \geq \ET_s$, we can do a binary search to find a good estimate of $R_s d_s$.

\begin{breakablealgorithm}
\caption{}
\begin{algorithmic}[1]
\REQUIRE
An upper bound $\bar\ET_s \geq \ET_s$. Initial estimate $\tilde{p}=1$.
\STATE
Construct the modified graph $\hat G$ with $\eta=\tilde{p}\cdot \bar\ET_s$, which satisfies $\frac{\hat\ET_\sigma}{\hat R_\sigma \hat d_\sigma} \leq 3$.
\STATE
Run the algorithm from \zcref{lem:eff-res-prim} on $\hat G$ with $\tau = 3$ and $T \in \Theta(\sqrt{\bar\ET_s})$ to get an estimate
\[
a
= \frac{1}{\sqrt{\hat R_\sigma \hat d_\sigma}} \pm \frac{1}{4\sqrt{\bar\ET_s}}
= \frac{1}{\sqrt{1 + \eta R_s d_s}} \pm \frac{1}{4\sqrt{\bar\ET_s}}.
\]
If $a \leq \frac{1}{2\sqrt{\bar\ET_s}}$ and $\tilde p \geq \frac{2}{\bar\ET_s}$, halve $\tilde p$ and go back to Step 1.
Otherwise, output $\tilde p$.
\ENSURE
$\tilde p \in \Theta(1/(R_s d_s))$ with constant probability.
\end{algorithmic}
\label{alg:fpaa}
\end{breakablealgorithm}

Assuming zero error probability in step 2, the algorithm terminates after $O(\log(R_s d_s))$ iterations.
The output $\tilde p$ is such that $\frac{1}{\sqrt{1 + 2 \tilde p \bar\ET_s R_s d_s}} - \frac{1}{4\sqrt{\bar\ET_s}} \leq \frac{1}{2\sqrt{\bar\ET_s}}$ while $\frac{1}{\sqrt{1 + \tilde p \bar\ET_s R_s d_s}} + \frac{1}{4\sqrt{\bar\ET_s}} \geq \frac{1}{2\sqrt{\bar\ET_s}}$, which implies that $7/18 \leq \tilde p R_s d_s \leq 16$.
Taking into account the error, we can apply the same argument as in \cite[Theorem 15]{brassard2000quantum} to argue that (i) the algorithm terminates after $O(\log(R_s d_s))$ iterations in expectation (and $\lceil\log(\bar\ET_s)\rceil$ iterations in the worst case), and (ii) the output satisfies the same guarantee with constant probability.
Combining this with the algorithm in \zcref{sec:eff-res-known} yields the following theorem.

\effres*

\begin{remark}
    The algorithms for effective resistance estimation and elf preparation assume an upper bound $\bar{\ET}_s\geq\ET_s$. Any known upper bound on $\HT_s$ is sufficient, since $\ET_s\leq\HT_s$.
    More generally, consider a connected $n$-vertex weighted graph, whose edge weights are bounded below by a known $w_{\min}>0$.
    A simple path from $s$ to $M$ gives $R_s\le(n-1)/w_{\min}$, and hence
    \[
    \mathrm{ET}_s \leq R_s W\leq \frac{(n-1) W}{w_{\min }}, 
    \]
    where $W=\sum_x d_x$.
    For an unweighted graph with at most $m$ undirected edges, the upper bound $2m(n-1)$ thus holds. 
    These bounds can be loose, and our complexity guarantees depend on the supplied bound.
    However, this is in correspondence to the earlier works \cite{piddock2019quantum,ito2019approximate,apers2022elfs}.
    For the special case of expanders, \zcref{thm:expander} also gives an upper bound in terms of just $n$ and the sink size $|M|$.
\end{remark}

\subsection{Lower bound}

We can easily prove that the error-scaling in \zcref{thm:intro-eff-res} is optimal.

\begin{lemma}
Fix $0 < \varepsilon < 1/2$.
Estimating $R_s d_s$ to multiplicative precision $<\varepsilon$ with constant success probability requires $\Omega(1/\varepsilon)$ quantum walk steps.
\end{lemma}
\begin{proof}
With $\delta \in \{\varepsilon/2,-\varepsilon/2\}$, consider a 3-vertex graph with vertices $\{s,x,t\}$, edges $\{(s,x),(x,t)\}$, and weights $\{w_{s,x} = 1/2+\delta,w_{x,t} = 1/2-\delta\}$.
A quantum walk on this graph corresponds to a reflection about the state $\ket{\phi_x} = \sqrt{1/2+\delta} \ket{xs} + \sqrt{1/2-\delta} \ket{xt}$.
Now if the sink $M = \{t\}$ then $R_s d_s = 1/(1/2-\delta)$, and an estimate of $R_s d_s$ to multiplicative precision $<\varepsilon$ would distinguish between $\delta = \varepsilon/2$ or $\delta = -\varepsilon/2$.
However, by \cite[Theorem 5.1]{belovs2025space}, solving this problem with constant success probability requires $\Omega(1/\varepsilon)$ reflections around~$\ket{\phi_x}$.
\end{proof}

\subsection{Witness size estimation} \label{sec:witness-size}

Generalizing beyond the graph setting, we can consider arbitrary projectors $\Pi,\Delta$ with an initial state $\ket{\psi} \in \ker(\Pi)$, and aim to estimate $\omega = 1/\| P_{\ker(\Pi)\cap\ker(\Delta)} \ket{\psi}\|^2$.
Using our effective gap transducer (\zcref{thm:eff-gap-transducer}), we get that the unitary $U = (2\Pi-I)(2\Delta-I)$ transduces $\ket{\psi} \overset{U}{\rightsquigarrow} (2P_{\ker(\Pi)\cap\ker(\Delta)}-I) \ket{\psi}$ with transduction complexity $W = \|(\Pi-\Pi \Delta \Pi)^+ \Delta \ket{\psi}\|^2$.
Following the same argument that led to \zcref{lem:eff-res-prim}, we can then get the following lemma.

\begin{lemma} \label{lem:witness-prim}
Given an upper bound $\tau \geq W+1$ and integer $T \geq 1$, there is a quantum algorithm that with probability at least $2/3$ returns an estimate $1/\sqrt{\omega} + O(1/T)$ while making $O(\tau T)$ calls to $U$ and $\widetilde{O}(\tau T)$ other elementary operations.
\end{lemma}

As we elaborate in \zcref{app:span-program}, this estimation problem exactly corresponds to the \emph{witness size estimation} problem for span programs, as was considered by Ito and Jeffery in \cite{ito2019approximate}.
In that case, $\Pi$ and $\ket{\psi}$ correspond to some fixed projector and state, $\Delta$ is a projector depending on some input string, $\omega$ is called the positive witness size, and $W$ is upper bounded by some measure of the maximum negative witness size.\footnote{In their notation, $\Pi = \Pi_{T^\perp}$, $\Delta = \Pi_{H(x)^\perp}$, $\ket{\psi} = \ket{w_0}$, $\omega = w_+(x)$ and $W$ is bounded by $\tilde W_-$.}
In their \cite[Theorem 10]{ito2019approximate} they give a rebalancing procedure, similar to our graph modification (\zcref{sec:graph-mod}), that allows us to balance between the transduction complexity $W$ and the quantity~$\omega$ that we want to estimate.
In particular, it allows us to implement a binary search strategy as in \zcref{subsec:eff-res-unknown} to find a multiplicative estimate.
This yields the following theorem, which improves on the $\tO(\sqrt{\omega W}/\varepsilon^{3/2})$ complexity in \cite[Theorem 11]{ito2019approximate}.

\begin{theorem} \label{thm:witness-size}
Consider projectors $\Pi,\Delta$ and an initial state $\ket{\psi} \in \ker(\Pi)$.
Given an upper bound $\sigma \geq \sqrt{\omega W}$, there is a bounded-error quantum algorithm that returns an $\varepsilon$-multiplicative estimate of~$\omega$ in
\[
O\left(\sigma \left(\frac{1}{\varepsilon} + \log(\omega)\right)\right)
\]
calls to $U$ and with a logarithmically larger time complexity.
\end{theorem}

Following the arguments in~\cite[Section 5]{ito2019approximate}, this improved witness size estimation algorithm also yields an improved algorithm for estimating the effective resistance of an unweighted graph in the adjacency matrix model.

\begin{corollary}
Given an upper bound $\mu \geq n\sqrt{R_s}$, there is a bounded-error quantum algorithm that returns an $\varepsilon$-multiplicative estimate of $R_s$ in an unweighted graph using
\[
O\left(\mu \left(\frac{1}{\varepsilon} + \log(R_s)\right)\right)
\]
quantum queries in the adjacency matrix model and with a logarithmically larger time complexity.
\end{corollary}

\section{High-precision elfs} \label{sec:high-prec-elfs}

We now turn to the task of actually preparing elfs.
In a first approach, we use fixed-point amplitude amplification \cite{yoder2014fixed} to obtain a finite-dimensional transducer that prepares elfs to high precision, proving \zcref{thm:elfs-informal} from the introduction.
As the algorithm in \cite{yoder2014fixed} requires partial rotations
\[
S_{f}(\theta)=I-(1-e^{i \theta}) \ketbra{f}{f} \quad \text{ and } \quad S_s(\phi)=I-(1-e^{-i \phi})\ketbra{\phi_s}{\phi_s},
\]
we will use the following lemma.
It follows from combining \zcref{thm:eff-gap-transducer-general} (with $\Pi=\Pi_*$ and $\Delta=\Pi_+$) and \zcref{thm:elf-transd-refl}.

\begin{lemma}[Elfs rotation]
\label{lem:elfs-rotation}
Define the operator
\[
U(\theta,\phi)
= S_s(\phi) \Big( I - (1-e^{i\theta})(I-\Pi_+) \Big)
\]
with parameters $\theta,\phi\in(0,\pi)$.
For any $\ket{\psi} \in \mathrm{span}\{\ket{\phi_s},\ket{f}\}$ it holds that
\[
\ket{\psi}
\overset{U(\theta,\phi)}{\rightsquigarrow}
S_s(\phi) S_{f}(\theta) \ket{\psi}
\]
with public space $\ker(\Pi_*)$ and transduction complexity $W = \frac{\ET_s}{R_s d_s} - 1$.
\end{lemma}

We can now prove the following theorem.
\begin{theorem} \label{thm:elfs-prep-fp}
Assuming a lower bound $\bar p \leq \frac{1}{R_s d_s}$, there is a transducer $\tilde U$ so that
\[
\ket{\phi_s}
\overset{\tilde U}{\rightsquigarrow} \ket{f} + O(\varepsilon)
\]
with transduction complexity
\[
O\left( \frac{1}{\sqrt{\bar p}} \frac{\mathrm{ET}_s}{R_s d_s}\log \frac{1}{\varepsilon} \right).
\]
The transducer $\tilde U$ makes one call to the quantum walk operator and $\tO(1)$ other elementary operations.
\end{theorem}
\begin{proof}
Given the lower bound $\bar p \leq \frac{1}{R_s d_s} = 2|\inner{f|\phi_s}|^2$, the fixed-point quantum search algorithm in \cite[Equation 11]{yoder2014fixed} describes a sequence of efficiently computable angles $\{\theta_\ell,\phi_\ell\}_{\ell=1}^L$ for $L \in O\left(\frac{1}{\sqrt{\bar p}} \log \frac{1}{\varepsilon} \right)$ such that
\[
S_s(\phi_L) S_f(\theta_L) \dots S_s(\phi_1) S_f(\theta_1) \ket{\phi_s}
= \ket{f} + O(\varepsilon).
\]
We can implement this sequence as a transducer by using \zcref{lem:contr-transd-comp} to sequentially compose~$L$ elfs rotations, as in \zcref{lem:elfs-rotation}.
\end{proof}

To get the complexity in \zcref{thm:elfs-informal}, we assume an upper bound $\bar\ET_s \geq \ET_s$ and a constant factor approximation $\frac{1}{2 R_s d_s} \leq \tilde p \leq \frac{1}{R_s d_s}$.\footnote{By using the median of means trick (see e.g.~\cite{vershynin2018high}), we can get such a constant-factor approximation with success probability $1-\delta$ in complexity $O(\sqrt{\bar\ET_s} \log(R_s d_s) \log(1/\delta))$ by applying \zcref{thm:intro-eff-res} $O(\log(1/\delta))$ times.}
We can then consider the modified graph $\hat G$ with $\eta = \tilde p \cdot \bar\ET_s$ (\zcref{sec:graph-mod}), which satisfies $\hat R_\sigma \hat d_\sigma = 1 + \eta R_s d_s \leq 1 + \bar\ET_s$.
Applying \zcref{thm:elfs-prep-fp} with $\bar p = \frac{1}{1+\bar\ET_s}$ then yields a transducer to prepare the electric flow to precision~$\varepsilon$ with transduction complexity
\[
O\left( \frac{1}{\sqrt{\bar p}} \frac{\hat\ET_\sigma}{\hat R_\sigma \hat d_\sigma} \log \frac{1}{\varepsilon} \right)
\in O\left( \sqrt{\bar\ET_s} \log \frac{1}{\varepsilon} \right).
\]

A technical caveat is that this prepares the electric flow $\ket{\hat f}$ from $\sigma$ to $M$ on $\hat G$, rather than $\ket{f}$ from $s$ to $M$ on $G$.
However, $\ket{\hat f}$ and $\ket{f}$ are proportional outside of the added edge $(\sigma s)$, and it holds that\footnote{Since $\ET_s \geq 1$, we can always choose our upper bound $\bar\ET_s \geq 4$ while losing at most a constant factor.
This ensures that the overlap is at least $1/2$.}
\begin{equation} \label{eq:mod-flow}
|\inner{\hat f|f}|^2
= \frac{R_s}{\hat R_\sigma}
= \frac{1}{1 + \frac{1}{\eta R_s d_s}}
\geq 1 - \frac{2}{\bar\ET_s}.
\end{equation}
In \zcref{sec:exact-elfs}, we construct a transducer that prepares $\ket f$ exactly with an additional garbage register.
This construction is used to implement the elfs process in \zcref{sec:elfs-process}, and we suspect it suffices for most applications.

\section{Exact elfs} \label{sec:exact-elfs}

We can prepare elfs exactly by using more advanced transducer machinery.
In particular, by allowing the use of an infinite counter, transducers can capture the complexity of Las Vegas algorithms that have unbounded worst-case complexity (both in the classical sense \cite{motwani1996randomized}, as we will use it, and in the quantum sense \cite{belovs2023one}).
In the following, we first describe a subroutine that makes a fixed number of calls to the elfs transducer, and that succeeds with constant probability if both an upper bound $\bar \ET_s \geq \ET_s$ is known, as well as an estimate of $R_s d_s$.
Then we show how to remove this last assumption.

\subsection{Bounded-length trial}

We will use the following lemma.

\begin{lemma}[Transducer Hadamard test] \label{lem:transd-had}
Assume that $V$ is a transducer with public space $\cal H$, transducer action $2\Pi-I$ and transduction complexity bounded by $W_V$.
Let $cV = I \otimes \ket{0}\bra{0} + V \otimes \ket{1}\bra{1}$ and $H$ the Hadamard gate.
Then for any $\ket{\psi} \in \cal H$, the operator $V' = (I \otimes H) cV (I \otimes H)$ transduces
\[
\ket{\psi} \ket{0}
\overset{V'}{\rightsquigarrow} \Pi \ket{\psi} \ket{0} + (I-\Pi) \ket{\psi} \ket{1}
\]
with transduction complexity $W_V/2$.
\end{lemma}
\begin{proof}
If $V(\ket{\psi} \oplus w) = (2\Pi-I) \ket{\psi} \oplus w$, then
\[
V'\left( \ket{\psi} \ket{0} \oplus \frac{1}{\sqrt{2}} w \ket{-} \right)
= \Pi\ket{\psi} \ket{0} + (I-\Pi) \ket{\psi} \ket{1} \oplus \frac{1}{\sqrt{2}} w \ket{-}. \qedhere
\]
\end{proof}

Using it, we can prove the following theorem.

\begin{theorem} \label{thm:bounded-length}
Fix parameters $B \geq 1$ and $\eta \geq 0$.
There exists a finite-dimensional transducer that transduces
\[
\ket{s} \ket{0}
\rightsquigarrow \ket{s} \left( c_s \ket{1} \ket{f_s} \ket{\Gamma_s} + \sqrt{1-|c_s|^2} \ket{0} \ket{f^\perp} \right)
\]
with $\ket{\Gamma_s},\ket{f^\perp}$ garbage states.
The transducer has transduction complexity
\[
O\left(\sqrt{B} \left(1 + \frac{\ET_s}{\eta R_s d_s}\right) \right)
\]
and uses one call to the quantum walk operator on $\hat G$ with parameter $\eta$ and $\tO(1)$ other elementary operations.
If $\eta R_s d_s \in [B,2B]$ then it has success probability $|c_s|^2 \geq 1/8$.
\end{theorem}
\begin{proof}
We set $L = \lceil 4\sqrt{B} \rceil$ and use the shorthand $z = \eta R_s d_s$.
The transducer first maps
\[
\ket{s}\ket{0}
\quad \rightsquigarrow \quad
\ket{s} \frac{1}{\sqrt{L}} \sum_{j=0}^{L-1} R^j \ket{\phi_\sigma} \ket{j} \ket{0}
\]
where the last transduction uses \zcref{lem:contr-transd-comp} on the rotation $R$ from \zcref{eq:rotation} on $\hat G$, similar to \zcref{lem:contr-transd}.
Now recall that $R^j \ket{\phi_\sigma} = \alpha_j \ket{\hat f} + \sqrt{1-\alpha_j^2} \ket{\hat f^\perp}$ with $\alpha_j = \sin((2j+1)\theta)$ and $\sin(\theta) = 1/\sqrt{2\hat R_\sigma \hat d_\sigma} = 1/\sqrt{2(1+z)}$.
We can hence compose with the transducer Hadamard test (\zcref{lem:transd-had}) with $V$ the elfs transducer on $\hat G$ to map
\[
R^j \ket{\phi_\sigma} \ket{0}
\rightsquigarrow \alpha_j \ket{\hat f}\ket{1} + \sqrt{1-\alpha_j^2} \ket{\hat f^\perp}
\]
Finally, we attach another qubit that indicates whether an edge is from the original graph $G$, mapping
\[
\ket{\hat f}\ket{0}
\rightarrow \gamma \ket{f} \ket{1} + \sqrt{1-\gamma^2} \ket{f^\perp} \ket{0},
\]
where $\gamma = \sqrt{z/(1+z)}$ (cfr.~\zcref{eq:mod-flow}).
By transducer composition this maps
\[
\ket{s}\ket{0}
\rightsquigarrow
\ket{s} \frac{\gamma}{\sqrt{L}} \sum_{j=0}^{L-1} \beta_j \ket{f} \ket{1}\ket{1} + \ket{\perp}
\]
with transduction complexity $O\left(\sqrt{B} \left(1 + \frac{\ET_s}{z}\right) \right)$ and while making 1 call to the quantum walk operator on $\hat G$.

It remains to argue that the success probability
\[
|c_s|^2
= \frac{\gamma^2}{L} \sum_{j=0}^{L-1} |\alpha_j|^2
\]
is at least $1/8$ when $z \in [B,2B]$.
To do so we first note that $\gamma^2 = z/(1+z) \geq 1/2$ when $z \geq B \geq 1$.
Then, we rewrite
\[
\frac{1}{L} \sum_{j=0}^{L-1} |\alpha_j|^2
= \frac{1}{L} \sum_{j=0}^{L-1} \sin^2((2j+1)\theta)
= \frac{1}{2} - \frac{\sin(4L\theta)}{4L\sin(2\theta)}.
\]
It remains to note that if $z \in [B,2B]$ then
\[
\sin(2\theta)
= \frac{\sqrt{1+2z}}{1+z}
\geq \frac{\sqrt{1+4B}}{1+2B}
\geq \frac{2}{3\sqrt{B}}
\]
and hence $|c_s|^2 \geq 1/4 - 3/64 \geq 1/8$.
\end{proof}

\begin{remark} \label{remark:local-estimates}
If we have an upper bound $\bar\ET_s \geq \ET_s$ and an estimate $\gamma \in [R_s d_s/2, R_s d_s]$, then we can set $B = \bar\ET_s$ and $\eta = \bar\ET_s/\gamma$ to yield transduction complexity $O(\sqrt{\bar\ET_s})$ and success probability $|c_s|^2 \geq 1/8$.
Using \zcref{lem:contr-transd-comp} to compose this transducer until successful results in a transducer with infinite counter that transduces $\ket{s}\ket{0} \rightsquigarrow \ket{f_s}\ket{\Gamma_s}$ with transduction complexity $O(\sqrt{\bar\ET_s})$, while making a single call to the quantum walk operator on $\hat G$.
\end{remark}

\subsection{Repeated sweep}

We now show how to avoid the estimate of $\eta$, at the cost of an additional log-factor.
The resulting transducer only takes $B$ as a parameter, and does not assume bounds on $\ET_s$ or $R_s d_s$.
At its core, the transducer encodes a Las Vegas algorithm that repeatedly runs \zcref{thm:bounded-length} for different value of $\eta$ until successful.

We use this to implement as a transducer the zero-error Las Vegas algorithm for quantum search from \cite{boyer1998tight}, the idea of which is to increase the number of calls to the amplitude amplification operator, until the algorithm succeeds.
The expected complexity of this algorithm is shown to scale optimally, and it is this expected complexity that transducers can capture.
Setting $B = \bar\ET_s/(1+\log(\bar\ET_s))^2$ for some upper bound $\bar\ET_s \geq \ET_s$ implies \zcref{thm:intro-elf-inf-bar}.

\begin{theorem} \label{thm:sweep}
Fix parameter $B \geq 1$.
There exists a transducer with an infinite counter that transduces
\[
\ket{s} \ket{0}
\rightsquigarrow \ket{s} \ket{f_s} \ket{\Gamma_s}
\]
with transduction complexity
\[
O\left(\sqrt{B} (1+\log(R_s d_s))^2 + \frac{\ET_s}{\sqrt{B}} \right).
\]
The transducer uses one call to the quantum walk operator on $\hat G$.
\end{theorem}
\begin{proof}
The transducer will do a sweep over different guesses $\eta_j = B 2^{-j}$.
Iterating over increasing $k = 1,2,\dots,$ it does the following: try $r = 16$ times the transducer from \zcref{thm:bounded-length} with parameter $B$ and guess $\eta_j$ for each $j \in \{0,1,\dots,k-1\}$.
Each try is implemented on a fresh register.
We compose this infinite sequence of transducers using \zcref{lem:contr-transd-comp}.

If $j = J \coloneq \lfloor \log_2(R_s d_s) \rfloor$ then the transducer from \zcref{thm:bounded-length} has success probability at least~$1/8$, and the probability of failing all $r$ trials is at most $q = (7/8)^r < 1/2$.
Every sweep with $k \geq J$ contains the guess $\eta_J$, and hence the probability of reaching stage $J + a$ is at most $q^a$.
This process will terminate with probability one, upon which it returns the flow state.

It remains to bound the transduction complexity.
By \zcref{thm:bounded-length} the transduction complexity $C_k$ at stage $k$ is
\[
C_k
\in O\left( \sum_{j=0}^k \sqrt{B} \left(1 + \frac{\ET_s 2^j}{B R_s d_s} \right) \right)
\in O\left( \sqrt{B} (k+1) + \frac{\ET_s 2^k}{\sqrt{B} R_s d_s} \right).
\]
Following \zcref{lem:contr-transd-comp}, the total transduction complexity follows by summing over the different stages, weighted by the survival probability.
This is
\[
\sum_{k=0}^J C_k + \sum_{a \geq 1} q^a C_{J+a}
\in O\left( \sqrt{B} (J+1)^2 + \frac{\ET_s}{\sqrt{B} R_s d_s} 2^J \right),
\]
where we used that $\sum_{a \geq 1} q^a (J + a + 1) \in O(J+1)$ and $\sum_{a \geq 1} q^a 2^a \in O(1)$ for $q < 1/2$.
\end{proof}

\section{Elfs process} \label{sec:elfs-process}

From our elfs transducers, we can build a transducer that implements the elfs process from \cite{apers2022elfs} to high precision and even exactly.
As described in \zcref{sec:elfs-intro}, the elfs process for an initial source vertex $s$ and a fixed sink $M$ describes a Markov chain $\{Y_0=s,Y_1,\dots,Y_\rho \in M\}$ over the vertices~$V$ with transition probabilities
\begin{equation} \label{eq:elfs-trans-prob}
\mathbb{P}(Y_{t+1}=y \mid Y_t=x)
= \sum_{z \in V} |\inner{y,z|f_x}|^2
= \frac{1}{2R_x} \sum_{z:(y,z)\in E} \frac{(f_x)_{yz}^2}{w_{yz}}
\end{equation}
with $f_x$ the unit electric flow from $x$ to $M$.
More verbosely, $y$ is a random endpoint of an edge sampled from $f_x$.
The process terminates when a vertex $Y_\rho \in M$ is drawn, and we call the electric hitting time $\EHT_s = \mathbb{E}[\rho]$ the expected time when this happens.
For comparison, we let $\{X_0 = s,X_1,\dots,X_\tau\in M\}$ denote a random walk from $s$ that is absorbed in $M$, so that $\mathbb{E}[\tau] = \HT_s$.
The following is a key fact.

\begin{lemma}[{\cite[Corollary 1]{apers2022elfs}}]
The elfs process has the same arrival distribution as the random walk, $Y_\rho \sim X_\tau$.
\end{lemma}

More strongly even, we can couple the two processes so that there are random stopping times $0 < \nu_1 < \dots < \nu_\rho = \tau$ such that $Y_i = X_{\nu_i}$.
This was illustrated in \zcref{fig:intro-coupling}.
From \cite[Lemma 2]{apers2022elfs} we know that
\[
\mathbb{E}[\nu_1] = \ET_s/2,
\]
which gives another operational interpretation of the random walk escape time $\ET_s$ as the expected number of random walk steps between the elfs.
It also implies the useful identity (\cite[Eq.~(16)]{apers2022elfs})
\begin{equation} \label{eq:exp-ET}
\mathbb{E}\left[\sum_{t=0}^{\rho-1} \ET_{Y_t}\right]
= 2 \HT_s.
\end{equation}

\subsection{Elfs process transducer}

Since this is the easiest case, we first assume that we are given upper bounds $\bar\ET_x$ and constant-factor approximations of $R_x d_x$ for all $x$ in the elfs process.
We can then invoke \zcref{thm:bounded-length} to simulate the elfs process exactly.
More precisely, the transducer uses an infinite counter to return a quantum sample of the elfs process $\{Y_i\}$
\[
\ket{h}
= \sum_{\{y_i\}} \sqrt{\mathrm{Pr}(\{y_i\})} \ket{\{y_i\}} \ket{\Gamma(\{y_i\})}
\]
where we let $\ket{\{y_i\}} = \ket{y_0,\dots,y_\rho,0,\dots}$ and $\ket{\Gamma(\{y_i\})}$ denotes garbage depending on $\{y_i\}$.
Measuring $\ket{h}$ effectively returns a full sample of the elfs process, and not just the endpoint.
Since the length $\rho$ of the elfs process is unbounded, this state requires an infinite counter to store.
Effectively though, the expected number of non-zero registers will be $\tO(\EHT_s)$, and this will allow us to turn this infinite-dimensional transducer into a finite-dimensional algorithm.

\begin{theorem}[Elfs process with local estimates]
\label{thm:algo for elfs process}
Given upper bounds $\bar\ET_x \geq \ET_x$ and constant-factor approximations of $R_x d_x$ for all $x$, there exists a transducer that outputs $\ket{h}$ exactly with transduction complexity of order
\[
\mathbb{E}\left[ \sum_t \sqrt{\bar\ET_{Y_t}} \right]
\]
and using 1 controlled call to the quantum walk operator and $\tO(1)$ elementary operations.
\end{theorem}
\begin{proof}
Similar to \zcref{thm:sweep}, we prove this using controlled transducer composition (\zcref{lem:contr-transd-comp}).
Using the notation
\begin{align*}
	\begin{aligned}
		&\psi_{0,0}\;\overset{S_0}{\rightsquigarrow}\;\\[-7pt]
		&\;\\[-7pt]
		&\;
	\end{aligned}
	\left\{
	\begin{aligned}
	&\psi_{1,0}\;\overset{S_1}{\rightsquigarrow}\;\\[-7pt]
	&\;{\scriptscriptstyle+}\\[-7pt]
	&\psi_{1,1}
\end{aligned}
	\right.
	\left\{
\begin{aligned}
	&\psi_{2,0}\;\overset{S_2}{\rightsquigarrow}\;\\[-7pt]
	&\;{\scriptscriptstyle+}\\[-7pt]
	&\psi_{2,1}
\end{aligned}
\right.
\begin{aligned}
	&\ldots\;\\[-7pt]
	&\;\\[-7pt]
	&\;
\end{aligned}
\end{align*}
we will set $\psi_{0,0} = \ket{s}\ket{0}$ and we will use the first register to record the elfs process, and the second register to store garbage.
As an inductive hypothesis (trivially satisfied for $t = 0$), assume that the transducer $S_t$ at time $t$ takes a superposition
\[
\psi_{t,0}
= \sum_{y_0,\dots,y_t \notin M} \sqrt{\mathrm{Pr}(y_0,\dots,y_t)} \ket{y_0,\dots,y_t} \ket{\Gamma_{y_0,\dots,y_t}}
\]
with $\mathrm{Pr}(\cdot)$ the probability measure of the first $t$ steps of the elfs process.
It holds that
\[
\| \psi_{t,0} \|^2
= \sum_{y_0,\dots,y_t \notin M} \mathrm{Pr}(y_0,\dots,y_t)
= \mathrm{Pr}(\rho > t).
\]
To implement a new step, $S_t$ applies the elfs transducer $V$ from \zcref{thm:bounded-length} (together with \zcref{remark:local-estimates}) on a state $\ket{y_0,\dots,y_t}\ket{\Gamma}$ to prepare the elfs state $\ket{f_{y_t}}$,
\[
\ket{y_0,\dots,y_t} \ket{\Gamma}
\overset{V}{\rightsquigarrow} \ket{y_0,\dots,y_t} \ket{f_{y_t},\Gamma'}.
\]
After that it copies the first register of $\ket{f_{y_t}} = \sum_z \alpha_z \ket{z,f_{y_t,z}}$ into a new register that is appended to the elfs register, simulating a deferred measurement:
\[
\sum_z \alpha_z \ket{y_0,\dots,y_t} \ket{z,f_{y_t,z},\Gamma'}
\rightarrow \sum_z \alpha_z \ket{y_0,\dots,y_t,z} \ket{z,f_{y_t,z},\Gamma'}.
\]
Notably, $|\alpha_z|^2 = \sum_{x \in V} |\inner{z,x|f_{y_t}}|^2$, so the new state $\psi_{t+1}$ satisfies
\[
\psi_{t+1}
= \sum_{y_0,\dots,y_t \notin M,z} \sqrt{\mathrm{Pr}(y_0,\dots,y_t,z)} \ket{y_0,\dots,y_t,z} \ket{\Gamma_{y_0,\dots,y_t,z}}.
\]
Letting $\psi_{t+1,0}$ (resp.~$\psi_{t+1,1}$) correspond to the component with $z \notin M$ (resp.~$z \in M$), we see that the new state satisfies the inductive hypothesis.

Given an upper bound $\bar\ET_{y_t} \geq \ET_{y_t}$ and a constant-factor estimate of $R_{y_t} d_{y_t}$, by \zcref{thm:bounded-length} the transduction complexity of $S_t$ on $\psi_{t,0}$ is of order
\[
\sum_{y_0,\dots,y_t \notin M} \mathrm{Pr}(y_0,\dots,y_t) \sqrt{\bar\ET_{y_t}}
= \mathbb{E}\left[\sqrt{\bar\ET_{Y_t}}\right].
\]
By \zcref{lem:contr-transd-comp}, the transduction complexity of the composition is of order $\mathbb{E}\Big[\sum_t \sqrt{\bar\ET_{Y_t}}\Big]$.
\end{proof}

\noindent
If our upper bounds are tight, $\bar\ET_x \in \Theta(\ET_x)$, then by Jensen's inequality the complexity is
\[
\mathbb{E}\Big[\sum_t \sqrt{\ET_{Y_t}}\Big]
\leq \sqrt{\EHT_s \cdot \mathbb{E}\Big[\sum \ET_{Y_t}\Big]}
= \sqrt{\EHT_s \cdot 2 \HT_s}.
\]

If we do not have local estimates of $\ET_x$ and~$R_x d_x$, then we can instead rely on the sweep transducer from \zcref{thm:sweep}.
Composing it identically to \zcref{thm:algo for elfs process} and using a fixed budget~$B$ for each step of the elfs process, the resulting transduction complexity is of the order
\begin{equation} \label{eq:elfs-proc-compl}
\mathbb{E}\left[ \sum_t \left( \sqrt{B} (1+\log(R_{Y_t} d_{Y_t}))^2 + \frac{\ET_{Y_t}}{\sqrt{B}} \right) \right]
= \sqrt{B} \, A_s + \frac{2\HT_s}{\sqrt{B}},
\end{equation}
where $A_s := \mathbb{E}\left[\sum_{t}(1+\log(R_{Y_t}d_{Y_t}))^2\right].$
While we lack a convenient uniform bound on $R_x d_x$, we can still prove a convenient bound on $A_s$.

\begin{lemma}
$A_s \leq \EHT_s \left( 2 + \log\left(\frac{2\HT_s}{\EHT_s} \right) \right)^2$.
\end{lemma}
\begin{proof}
For $x\notin M$, define
\[
    \mu(x):=
    \frac{\mathbb{E}[|\{t<\rho:Y_t=x\}|]}{\EHT_s}.
\]
Then $\mu$ is a probability distribution on $V\backslash M$, and it allows us to rewrite
\[
A_s
= \EHT_s \sum_{x\notin M} \mu(x) \big( 1+\log(R_x d_x) \big)^2. 
\]
Since $(2+\log z)^2$ is increasing and concave for $z\ge1$,
Jensen's inequality gives
\[
A_s
\le \EHT_s \sum_{x\notin M}\mu(x)(2+\log(R_xd_x))^2 \\
\le \EHT_s
\left(
    2+\log\left(\sum_{x\notin M}\mu(x)R_xd_x\right)
\right)^2.
\]
It remains to note that $R_x d_x \leq \ET_x$ and that  by \zcref{eq:exp-ET} we have
\[
\sum_{x\notin M} \mu(x)\ET_x
= \frac{1}{\EHT_s}
  \mathbb{E}\left[\sum_{t=0}^{\rho-1}\ET_{Y_t}\right]
= \frac{2\HT_s}{\EHT_s}. \qedhere
\]
\end{proof}

This implies the following theorem.

\begin{theorem}[Elfs process with global estimate] \label{thm:elfs-process-global}
Fix parameter $B \geq 1$.
There exists a transducer that outputs $\ket{h}$ exactly with transduction complexity of order
\begin{equation} \label{eq:elfs-global}
\sqrt{B} \, \EHT_s \, \left( 1 + \log\left(\frac{2\HT_s}{\EHT_s} \right) \right)^2
+ \frac{\HT_s}{\sqrt{B}}.
\end{equation}
\end{theorem}

The transducer requires no local parameter estimates.
Rather, it suffices to have upper bounds on $\HT_s$ and $\EHT_s$ as we show in the following lemma.
This establishes \zcref{thm:intro-elfs-process} from the introduction.

\begin{lemma}\label{lem:para-choice}
Let $\bar H \geq \HT_s$ and $2 \bar H \geq \bar h \geq \EHT_s$.
If
\[
B
= \max\left\{ 1 , \frac{\bar H/\bar h}{\left(1+\log\left(2(1+\bar H/\bar h)\right)\right)^2} \right\}
\]
then \zcref{eq:elfs-global} is $O\left(\sqrt{\bar H \bar h} \log(1+\bar H/\bar h)\right)$.
\end{lemma}
\begin{proof}
Write $H=\mathrm{HT}_s$, $h=\mathrm{EHT}_s$, and define $r=\bar H/\bar h\ge\frac12$ and $L=1+\log(2(1+r))$ so that $B=\max\{1,r/L^2\}$.

We first establish
\[
h\left(1+\log\frac{2H}{h}\right)^2
=O(\bar hL^2).
\]
Set $u=h/\bar h\in(0,1]$. Since $h\le2H$ and
$H\le\bar H$, we have
\[
0<1+\log\frac{2H}{h}
\le L+\log\frac1u.
\]
Using $(a+b)^2\le2a^2+2b^2$ gives
\[
h\left(1+\log_2\frac{2H}{h}\right)^2
\le \bar h\,u\left(L+\log\frac1u\right)^2
\le 2\bar hL^2
   +2\bar h\,u \log^2(u)
\in O(\bar hL^2),
\]
where we used that $L \ge 1$ and $u \log^2(u) \leq 2$ for $0<u\leq 1$.

Consequently, \zcref{eq:elfs-global} is
\[
O\!\left(\sqrt B\,\bar hL^2+\frac{\bar H}{\sqrt B}\right).
\]
Writing $Q=\sqrt{\bar H\bar h}\,L$, the choice of $B$ gives
\[
\sqrt B\,\bar hL^2
=Q\max\left\{1,\frac{L}{\sqrt r}\right\},
\qquad
\frac{\bar H}{\sqrt B}
=Q\min\left\{1,\frac{\sqrt r}{L}\right\}.
\]
Finally, $\ln(1+r)\le\sqrt r$ and $r\ge1/2$ imply
\[
\frac{L}{\sqrt r}
=\frac{2+\log_2(1+r)}{\sqrt r}
\le2\sqrt2+\frac1{\ln2}.
\]
Both terms are therefore $O(Q)$, proving the claim.
\end{proof}



While it is unclear whether we can avoid these upper bounds on $\HT_s$ and $\EHT_s$, in the following applications we show that in certain cases we can either circumvent them, or use some uniform upper bounds.

\subsection{Application 1: Quantum walk search}

Consider the following graph search problem: starting from an initial vertex $s$, find a ``marked'' element from the sink set $M$.
In \cite{apers2022elfs}, a quantum walk search algorithm for this problem was described based on the elfs process.
They showed in \cite[Theorem 4]{apers2022elfs} that, if there is a quantum walk algorithm that simulates the elfs process (and hence solves the search problem) given upper bounds $\bar H \geq \HT_s$ and $\bar h \geq \EHT_s$ in $f(\bar H,\bar h)$ many quantum walk steps, then there is a quantum search algorithm making $\tO(f(\HT_s,\EHT_s))$ many quantum walk steps \emph{without} using any upper bounds.
The key point is that the search setting allows us to \emph{verify} when an algorithm was successful, and so we can effectively do a scan over increasing estimates.
The tradeoff is that such a strategy no longer carries guarantees about the arrival distribution.

While \cite{apers2022elfs} has a bound $f(\bar H,\bar h) \in \tO(\sqrt{\bar H} \bar h^{3/2})$, our \zcref{thm:intro-elfs-process} improves this to $f(\bar H,\bar h) \in \tO(\sqrt{\bar H \bar h})$.
By the same sweeping strategy as in \cite[Theorem 4]{apers2022elfs}, this thus implies the following theorem.

\begin{theorem}[Quantum walk search]
Given a source vertex $s$ on a graph, and query access\footnote{We make queries of the form ``Is $x \in M$''? No other prior knowledge of $M$ is required.} to the vertices in the sink $M$.
There is a quantum walk algorithm that returns an element from the sink $M$ using a number of quantum walk steps
\[
\tO(\sqrt{\HT_s \EHT_s}).
\]
\end{theorem}

\subsection{Application 2: Elfs process on expanders}

As our second application, we develop the elfs process transducer for the particular case of expander graphs, which are a useful proxy for random and real-world networks.
Here we show uniform bounds on all relevant quantities, which allows us to invoke our elfs process transducer with local estimates, \zcref{thm:algo for elfs process}.
This avoids the logarithmic overheads in \zcref{thm:elfs-process-global}.
The following theorem is proven in \zcref{app:expander}.

\begin{restatable}{theorem}{expander} \label{thm:expander}
If $G$ is a regular, bounded-degree expander graph with a sink $M \subseteq V$ of size $|M| = m$, then for any $x \in V$ we have
\[
R_x \in \Theta(1),
\quad
\ET_x \in O\left(1 + \frac{n}{m^2}\right),
\quad
\EHT_x \in O\left(\min\left\{m,\frac{n}{m} + \log n\right\}\right),
\quad
\HT_x \in O\left( \frac{n}{m} + \log n \right).
\]
\end{restatable}

Given these bounds, we can apply our exact elfs process transducer from \zcref{thm:algo for elfs process} using only knowledge of $n$ and $m$.
By uniformly setting our upper bounds $\bar\ET_{Y_t} = \bar\ET \in O(1+n/m^2)$, we get a transducer that samples exactly from the random walk arrival distribution on $M$ with transduction complexity bounded by
\[
\mathbb{E}\left[ \sum_t \sqrt{\bar\ET_{Y_t}} \right]
= \EHT_s \sqrt{\bar\ET}
\in O(\min\{\sqrt{n},\HT_s\}).
\]
If $m \ll \sqrt{n}$ then this yields a quantum speedup over the random walk hitting time $\HT_s \in O(n/m + \log n)$.
Turning the resulting transducer into a bounded-error quantum algorithm, we get the following theorem.

\begin{theorem}[Elfs on an expander] \label{thm:elfs-expander}
Let $G$ be a regular, bounded-degree expander graph on $n$ vertices with sink $M \subseteq V$.
For an initial vertex $s$, let $\{p_x\}_{x \in M}$ denote the arrival distribution of a random walk from $s$ to $M$.
There exists a quantum algorithm that returns a state
\[
\ket{\tilde h}
= \sum_x \sqrt{\tilde p_x} \ket{x} \ket{\Gamma_x}
\]
with garbage $\ket{\Gamma_x}$ so that $\| p - \tilde p \|_{TV} \leq \varepsilon$.
The quantum algorithm uses $O(\sqrt{n}/\varepsilon^2)$ quantum walk steps, $\tO(\sqrt{n}/\varepsilon^2)$ other elementary operations, and $\tO(\sqrt{n})$ space.
\end{theorem}

As discussed in \zcref{sec:elfs-learning}, this quantum algorithm has applications in semi-supervised learning on expander graphs.

This application illustrates the benefit of our improved dependence on the electric hitting time. For $1\le m\le\sqrt n$, substituting the bounds above into the complexity bound of~\cite{apers2022elfs} gives $\widetilde O(\sqrt n\,m/\varepsilon^2)$ quantum walk steps for sampling from the random walk arrival distribution. Our implementation of the same elfs process uses $O(\sqrt n/\varepsilon^2)$ quantum walk steps. Thus, when $m$ grows polynomially with $n$, our bound is polynomially smaller.

\section{Acknowledgements}

YZ was supported by the National Key R\&D Program of China (Grant No. 2025YFA1017200), and the 2026 International Research Collaboration Scholarship from AMSS, CAS.
Part of this work was done while YZ was visiting QuIC at Université libre de Bruxelles, hosted by Jérémie Roland. YZ is grateful to Simon Apers and Jérémie Roland for their hospitality and support.
SA was supported in part by the European QuantERA project QOPT (ERA-NET Cofund 2022-25), the French PEPR integrated projects EPiQ (ANR-22-PETQ0007) and HQI (ANR-22-PNCQ-0002), and the French ANR project QUOPS (ANR-22-CE47-0003-01).

\textbf{AI disclosure (v1):} OpenAI's GPT‑5.4 Thinking provided valuable insights for the proof of \zcref{thm:expander}.
The authors put together the actual proof, and verified the correctness and originality of all content.

\textbf{AI disclosure (v2):} Updated version introduces an improved sweeping strategy in \zcref{sec:exact-elfs,sec:elfs-process}, proposed by OpenAI's GPT-6 Astra, to avoid the need for local estimates.

\bibliographystyle{alpha}
\bibliography{elfs.bib}

\begin{appendices}

\section{Span program witness size estimation} \label{app:span-program}

Here we connect our claims on witness size estimation (\zcref{lem:witness-prim,thm:witness-size}) to the problem of estimating the witness size of a span program.

\subsection{Span program background}

We recall only the minimum background needed for this application, and refer the reader to~\cite{ito2019approximate} for details. 
\begin{definition}[{Span program, \cite[Definition 2]{ito2019approximate}}]
A span program on $[q]^n$ is specified by a tuple $P=(H, V, \tau, A)$, that consists of finite-dimensional inner product spaces
\[
H = H_1 \oplus \cdots \oplus H_n \oplus H_{\mathrm{true}} \oplus H_{\mathrm{false}}
\]
and $\{H_{j, a} \subseteq H_j\}_{j \in[n], a \in[q]}$ such that $H_{j, 1}+\cdots+H_{j, q}=H_j$, a vector space $V$, a target vector $\tau \in V$, and a linear operator $A \in \mathcal{L}(H, V)$. For each input $x \in[q]^n$, we associate a subspace $H(x):=H_{1, x_1} \oplus \cdots \oplus H_{n, x_n} \oplus H_{\text{true }}$.
\end{definition}

A span program partitions $[q]^n$ into two sets.
Positive inputs, denoted $P_1$, correspond to the case where $\tau \in \Im(A \Pi_{H(x)})$.
The remaining inputs, denoted $P_0 = [q]^n \backslash P_1$, are called the negative inputs.
We restrict our attention to positive inputs $x \in P_1$.

\begin{definition}[{Positive witness, \cite[Definition 3]{ito2019approximate}}]
A \emph{positive witness} $w_x$ for $x \in P_1$ is a vector
\(w_x \in H(x)\) such that \(A w_x = \tau\), and its minimum squared norm is the \emph{positive witness size}
\[
w_+(x)
\coloneqq \min \bigl\{ \|w_x\|^2 \mid w_x \in H(x),\ A w_x =\tau \bigr\}.
\]
\end{definition}

\begin{definition}[{Negative error, \cite[Definition 5]{ito2019approximate}}]
The negative error for $x \in P_1$ is defined as
\begin{equation} \label{eq:e-min}
e_-(x)
\coloneqq \min\bigl\{ \| \Pi_{H(x)} A^\dag \omega \|^2 \mid \omega \in V, \; \tau^\dag \omega = 1 \bigr\},
\end{equation}
with $\cdot^\dag$ the adjoint.
The optimal min-error negative witness size is then
\[
\tilde w_-(x)
\coloneqq \min\bigl\{ \| A^\dag \omega \|^2 \mid \omega^\dag \tau = 1, \; \| \Pi_{H(x)} A^\dag \omega \|^2 = e_-(x) \bigr\},
\]
and its maximum is denoted $\tilde W_- \coloneqq \max_{x \in P_1} \tilde w_-(x)$.
\end{definition}

What is shown in \cite[Theorem 6]{ito2019approximate} is that there exists a quantum algorithm that estimates the positive witness size $w_+(x)$ on input $x$ with multiplicative error $\eps$ using
\[
O\left(\sqrt{w_+(x)\,\widetilde W_-}/\varepsilon^{3/2}\right)
\]
queries to $x$.

\subsection{Witness size estimation}

Here we connect the estimation problem that we solve in \zcref{sec:witness-size} to the problem of estimating the witness size of a span program, as considered in \cite{ito2019approximate}.
There, they consider an initial state
\[
\ket{w_0}
= A^+ \tau,
\]
where it is assumed that $\tau$ is scaled in such a way that $\ket{w_0}$ has unit norm (this corresponds to a ``normalized'' span program \cite[Definition 8]{ito2019approximate}).
Now define $T:=\ker A \oplus \mathrm{span}\{|w_0\rangle\}$ and $Q_x:=T\cap H(x)$, and note that $Q_x$ is the positive-witness subspace for $x$.
In particular, it holds that
\[
\|P_{Q_x}|w_0\rangle\|^2=\frac{1}{w_+(x)},
\]
which is the quantity we want to estimate.
Next we show that the quantity $\tilde w_-(x)$ can be reexpressed using pseudo-inverses.
\begin{lemma} \label{lem:transd-span}
It holds that
\[
\tilde w_-(x)
= 1 +
\left\| \big( \Pi_{T^\perp} \Pi_{H(x)} \Pi_{T^\perp} \big)^+ \Pi_{H(x)^\perp} \ket{w_0} \right\|^2.
\]
\end{lemma}
\begin{proof}
First, notice that any $\omega$ minimizing $e_-(x)$ in \zcref{eq:e-min} satisfies $A^\dag \omega = \ket{w_0} + \nu$ with $\nu \in T^\perp$.
Moreover, since $A^\dag \omega$ minimizes $\| \Pi_{H(x)} A^\dag \omega \| = \| \Pi_{H(x)} (\ket{w_0} + \nu) \|$ over $\nu \in T^\perp$, $\nu$ must satisfy the equation
\[
\Pi_{T^\perp} \Pi_{H(x)} ( \ket{w_0} + \nu )
= 0.
\]
Equivalently, it must hold that $\Pi_{T^\perp} \Pi_{H(x)} \Pi_{T^\perp} \nu = -\Pi_{T^\perp} \Pi_{H(x)} \ket{w_0}$, and so the minimal norm solution will satisfy
\[
\nu
= -(\Pi_{T^\perp} \Pi_{H(x)} \Pi_{T^\perp})^+ \Pi_{T^\perp} \Pi_{H(x)} \ket{w_0}
= (\Pi_{T^\perp} \Pi_{H(x)} \Pi_{T^\perp})^+ \Pi_{H(x)^\perp} \ket{w_0},
\]
where the last equality used that $\Pi_{T^\perp} \ket{w_0} = 0$.
Using that $\tilde w_-(x) = \| A^\dag \omega \|^2 = 1 + \| \nu \|^2$ completes the proof.
\end{proof}

We can now make precise the connection between our \zcref{thm:witness-size} and the setting of witness size estimation in \cite{ito2019approximate}.
Specifically, in \zcref{lem:witness-prim,thm:witness-size} we can set $\Pi = \Pi_{T^\perp}$, $\Delta = \Pi_{H(x)^\perp}$ and $\ket{\psi} = \ket{w_0}$.
We then get that
\[
\| P_{\ker(\Pi) \cap \ker(\Delta)} \ket{\psi} \|^2
= \| P_{Q_x} \ket{w_0} \|^2
= 1/w_+(x).
\]
Finally, by \zcref{lem:transd-span}, the transduction complexity $W = \| (\Pi-\Pi\Delta\Pi)^+ \Delta \ket{\psi} \|^2$ is equal to $\tilde w_-(x) - 1$.
Our \zcref{thm:witness-size} then shows that we can estimate $w_+(x)$ to multiplicative error $\varepsilon$ while making
\[
O\left(\sqrt{w_+(x)\,\widetilde W_-}/\varepsilon\right)
\]
calls to $U = (2\Pi-I)(2\Delta-I) = (2\Pi_T-I)(2\Pi_{H(x)}-I)$.
By \cite[Lemma 2]{ito2019approximate}, we can implement $U$ with only 2 queries to $x$, and so we improve the error scaling in \cite{ito2019approximate} from $1/\varepsilon^{3/2}$ to $1/\varepsilon$.

\section{Hitting and escape time on expanders} \label{app:expander}

Here we prove useful bounds on the relevant random walk quantities on an expander.

\expander*

\begin{proof}
Let $P$ denote the random walk transition matrix on $G$, $\delta \in \Omega(1)$ its spectral gap, and $d$ the degree of $G$.
Let $Q$ denote the submatrix of $P$ with rows and columns in $V \backslash M$.
A standard but extremely useful fact is that for a random walk starting from $s$, it holds that
\[
\mathbb{E}_s[\text{$\#$ visits to $x$ before $\tau_M$}]
= \sum_{t \geq 0} Q^t_{sx}
= d v_x.
\]
See for instance \cite[Proposition 2.1]{lyons2017probability}.
This implies that $R_s d = \sum_t Q^t_{ss}$, $\HT_s = \sum_{x,t} Q^t_{sx}$ and
\[
\ET_s
= \frac{1}{d R_s} \sum_{x,t,t'} Q^t_{sx} Q^{t'}_{sx}
= \frac{1}{d R_s} \sum_{t,t'} Q^{t+t'}_{ss}
= \frac{1}{d R_s} \sum_t (t+1) Q^t_{ss}.
\]
We can bound both quantities by combining the bound
\begin{equation} \label{eq:Q-bnd-sp}
Q^t_{sx}
\leq (1-\delta m/n)^t,
\end{equation}
which follows from the fact that $\| Q \| \leq 1 - \delta m/n$ \cite[Lemma 18.2]{childs2017lecture}, with the bound
\begin{equation} \label{eq:Q-bnd-mix}
Q^t_{sx}
\leq P^t_{sx}
\leq \frac{1}{n} + (1-\delta)^t,
\end{equation}
which follows from the fact that $\| P - u u^T/n \| \leq 1-\delta$, with $u$ the all-ones vector.
Combining these bounds we get that
\begin{equation} \label{eq:Q-comb-bnd}
\sum_x Q^{t+r}_{sx}
= \sum_{x,y} Q^t_{sy} Q^r_{yx}
\leq \left(\frac{1}{n} + (1-\delta)^t \right) \sum_{x,y} Q^r_{yx}
\leq \left(1 - \frac{\delta m}{n} \right)^r + n (1-\delta)^t,
\end{equation}
where we used that $\sum_{x,y} Q^r_{yx} \leq \inner{\vec{1},Q^r\vec{1}} \leq n \|Q\|^r$.
This yields the bound
\begin{align*}
\HT_s
= \sum_{x,t} Q_{sx}^t
&\leq O(\log(n)) + \sum_{x,t \in \Omega(\log n)} Q_{sx}^t \\
&\leq \log(n) + \sum_{t\in\Omega(\log n)} \left(1 - \frac{\delta m}{n} \right)^{\lfloor t/2\rfloor} + n (1-\delta)^{\lceil t/2\rceil}
\in O\left( \log(n) + \frac{n}{\delta m} \right).
\end{align*}
To bound $\ET_s = \frac{1}{d R_s} \sum_t (t+1) Q^t_{ss}$ we again split the sum around $\log(n)$.
For the first terms we have that
\begin{equation} \label{eq:bound-small}
\sum_{t\in O(\log n)} (t+1) Q^t_{ss}
\overset{\eqref{eq:Q-bnd-mix}}{\leq}
\sum_{t\in O(\log n)} (t+1) \left( \frac{1}{n} + (1-\delta)^t \right)
\in O(1),
\end{equation}
where the last inclusion uses that $\sum_{t \geq 0} (t+1) (1-\delta)^t = 1/\delta^2 \in O(1)$.
For the larger terms we rewrite
\begin{equation} \label{eq:bound-large}
\begin{aligned}
Q^{u+t+r}_{ss}
= \sum_x Q^u_{sx} Q^{t+r}_{xs}
&\overset{\eqref{eq:Q-bnd-mix}}{\leq} \left(\frac{1}{n} + (1-\delta)^u\right) \sum_x Q^{t+r}_{xs} \\
&\overset{\eqref{eq:Q-comb-bnd}}{\leq} \left(\frac{1}{n} + (1-\delta)^u\right) \left( \left(1 - \frac{\delta m}{n} \right)^r + n (1-\delta)^t \right).
\end{aligned}
\end{equation}
We can then bound
\begin{align*}
\sum_{t \in \Omega(\log n)} (t+1) Q^t_{ss}
&\leq \sum_{t \in \Omega(\log n)} (t+1) \left(\frac{1}{n} + (1-\delta)^{\lfloor t/3\rfloor}\right) \left( \left(1 - \frac{\delta m}{n} \right)^{\lfloor t/3\rfloor} + n (1-\delta)^{\lfloor t/3\rfloor} \right) \\
&\leq \frac{2}{n} \sum_{t \in \Omega(\log n)} (t+1) \left( \left(1 - \frac{\delta m}{n} \right)^{\lfloor t/3\rfloor} + n (1-\delta)^{\lfloor t/3\rfloor} \right) \\
&\in O\left( \frac{n}{m^2} + 1 \right),
\end{align*}
where the last inclusion repeatedly uses that $\sum_{t \geq 0} (t+1) (1-p)^t = 1/p^2$.
By combining the bounds in \zcref{eq:bound-small} and \zcref{eq:bound-large} we also get the bound $R_s = \frac{1}{d} \sum_t Q^t_{ss} \in O(1)$.

It remains to bound the electric hitting time $\EHT_x$, for which we reiterate the argument from~\cite{apers2022elfs}.
The bound $\EHT_s \in O(n/m + \log n)$ follows from the fact that $\EHT_s \leq 2\HT_s$ (\zcref{eq:exp-ET}).
The $\EHT_s \in O(m)$ bound follows from the fact that the electric flow must push a unit flow through at most $d m$ edges going into $M$.
Since the effective resistance is $O(1)$, this implies that measuring the electric flow returns such an edge (and hence an element in $M$) with probability $\Omega(1/m)$.
The elfs process hence terminates after $O(\min\{m,n/m + \log n\})$ samples in expectation.
\end{proof}

\end{appendices}

\end{document}